\documentclass{jfmX}%\documentclass[12pt]{article}
\usepackage{amssymb,amsmath,amsfonts}
\usepackage{graphicx,graphics}
\DeclareGraphicsExtensions{.jpg,.jpeg,.pdf,.png,.mps,.eps,.ps,.gif}
\usepackage{natbib}
\usepackage{color}
\input colors.defs

\def\DEL#1{}% suggested deletions off
\def\AD#1{} % maybe change off
\def\RoraiQ#1{} % addition
\def\RMKA#1{}         % addition

\usepackage{framed}
\usepackage{multicol}

\newcommand{\bminil}[1]{\begin{minipage}[l]{#1 \textwidth}}
\newcommand{\bminir}[1]{\begin{minipage}[r]{#1 \textwidth}}
\newcommand{\bminic}[1]{\begin{minipage}[c]{#1 \textwidth}}
\newcommand{\emini}{\end{minipage}}
\newcommand{\EQL}{\begin{equation}\label}
\newcommand{\EQ}{\begin{equation}}
\newcommand{\EN}{\end{equation}}
\newcommand{\BFG}{\begin{figure}}
\newcommand{\EFG}{\end{figure}}
\newcommand{\ITM}{\begin{itemize}}
\newcommand{\ITN}{\end{itemize}}
\newcommand{\ENM}{\begin{enumerate}}
\newcommand{\EEN}{\end{enumerate}}
\newcommand{\BEA}{\[\begin{array}}
\newcommand{\EEA}{\end{array}\]}
\newcommand{\EQAL}{\begin{eqnarray}\label}
\newcommand{\EQA}{\begin{eqnarray}}
\newcommand{\ENA}{\end{eqnarray}}

\newcommand{\bB}{\mbox{\boldmath$B$}}

\newcommand{\bD}{\mbox{\boldmath$D$}}

\newcommand{\bN}{\mbox{\boldmath$N$}}

\newcommand{\bT}{\mbox{\boldmath$T$}}

\newcommand{\br}{\mbox{\boldmath$r$}}
\newcommand{\bs}{\mbox{\boldmath$s$}}

\newcommand{\bu}{\mbox{\boldmath$u$}}
\newcommand{\bv}{\mbox{\boldmath$v$}}

\newcommand{\bx}{\mbox{\boldmath$x$}}

\newcommand{\bomega}{\mbox{\boldmath$\omega$}}

\newcommand{\p}{\partial}

\newcommand{\dppt}{{\displaystyle\frac{\partial}{\partial t}}}

\newcommand{\ppt}{\frac{\partial}{\partial t}}

\newcommand{\ddt}{\frac{d}{dt}}

\newcommand{\dddt}{{\displaystyle\frac{d}{d t}}}

\renewcommand{\theenumi}{\alph{enumi}}

\newcommand{\bbR}{\mathbb{R}}

\newcommand{\biband}{\&~}
\newcommand\etall{\mbox{\textit{et al.}}}
\newcommand{\authone}[2]{{\sc #2,~#1}}
\newcommand{\authtwo}[4]{{\sc #2,~#1,~\&~#4,~#3}}
\newcommand{\auththr}[6]{{\sc #2,~#1,~#4,~#3,~\&~#6,~#5}}
\newcommand{\authfour}[8]{{\sc #2,~#1,~#4,~#3,~#6,~#5,~\&~#8,~#7}}
\newcommand{\authmanytwo}[4]{{\sc #2,~#1,~#4,~#3,}}
\newcommand{\authmanythr}[6]{{\sc #2,~#1,~#4,~#3,~#6,~#5,}}
\newcommand{\authmanyfour}[8]{{\sc #2,~#1,~#4,~#3,~#6,~#5,~#8,~#7,}}

\newcommand{\yjour}[6]{ #1~ #6 {\em #2} {\bf #3}, #4#5.}
\newcommand{\yproc}[7]{ #1~ #4. In {\em #5} (ed. #6), pp. #2-#3. #7.}

\newcommand{\ybook}[3]{ #1~ {\em #2}. #3.}
\newcommand{\yweb}[3]{ #1~ #2 #3.}

\numberwithin{figure}{section}
\numberwithin{equation}{section}
\newcommand{\mycomment}[1]{} %(OFF)
\newcommand{\bibcomment}[1]{}%(OFF)
\newcommand{\rem}[1]{}
\DeclareMathAlphabet{\mathbi}{OML}{cmm}{b}{it}

\newtheorem{proposition}{Proposition}

\newtheorem{corollary}{Corollary}

\makeatletter
\@addtoreset{equation}{section}
\def\theequation{\thesection.\arabic{equation}}
\makeatother

\title[\vspace{-4mm} Separation of quantum vortices]{Approach and separation of quantum vortices with balanced cores}
\author{C. Rorai$^1$, J. Skipper$^2$, R. M. Kerr$^2$, K. R. Sreenivasan$^3$}
\affiliation{$^1$Department of Engineering, University of Cambridge, Trumpington Street, Cambridge, UK  CB2 1PZ ceciliarorai@gmail.org\\[\affilskip]
$^2$Department of Mathematics, University of Warwick, Coventry, UK CV4 7AL\\[\affilskip]
$^3$Departments of Physics, Mechanical and Aerospace Engineering, and the Courant Institute of Mathematical Sciences, New York University, Bobst Library, 70 Washington Square South, New York City, NY 10012}
\pubyear{2016}
\volume{800}
\pagerange{xxx--xxx}
\begin{document}
\maketitle
%\centerline{http://arxiv.org/abs/1410.1259}
\begin{abstract}
The scaling laws of isolated quantum vortex reconnection are characterised by numerically 
integrating  the three-dimensional Gross-Pitaevskii equations, the simplest mean-field 
equation for a quantum fluid. The {\it primary} result is the identification of distinctly 
different  temporal power laws for the pre- and post-reconnection separation distances 
$\delta(t)$ for two configurations. For the initially anti-parallel case, the scaling laws 
before and after the reconnection time $t_r$ obey the dimensional 
$\delta\sim|t_r-t|^{1/2}$ prediction with temporal symmetry about $t_r$ and physical 
space symmetry about the mid-point between the vortices $x_r$. The extensions of the 
vortex lines close to reconnection form the edges of an equilateral pyramid. For all of 
the initially orthogonal cases, $\delta\sim|t_r-t|^{1/3}$ before reconnection and 
$\delta\sim|t-t_r|^{2/3}$ after reconnection, which are respectively slower and faster 
than the dimensional prediction. For both configurations, smooth scaling laws are 
generated due to two innovations. One is to initialise with density profiles about 
the vortex cores that 
%areslightly below the usual two-dimensional steady-state Pad\'e approximate profiles.This 
suppress unwanted secondary temporal density fluctuations. 
The other innovation is the accurate identification of the position of the vortex cores 
from a pseudo-vorticity constructed on the three-dimensional grid from the gradients of 
the wave function. These trajectories allow us to calculate the Frenet-Serret frames and 
the curvature of the vortex lines, {\it secondary} results that might hold clues for 
the origin of the differences between the scaling laws of the two configurations. 
For the orthogonal cases, the reconnection takes place in a 
{\it reconnection} plane defined by the directions of the curvature and vorticity. 
To characterise the structure further, lines are drawn that connect the four arms that 
extend from the reconnection plane, from which four angles $\theta_i$ between the lines 
are defined. Their sum is convex or hyperbolic, that is
$\sum_{i=1,4}\theta_i>360^\circ$, 
as opposed to the acute angles of the pyramid found for the anti-parallel initial conditions.
\end{abstract}
% \pacs{47.37.+q,47.27.De,47.32.C-,67.25.dk}
\vspace{-6mm} 
\section{Background\label{sec:back}}

The term ``quantum turbulence'' refers to a tangle of quantum vortex lines, a tangle whose formation and decay is determined by how these vortices collide, reconnect and separate. 
Although superfluid tangles form in a variety of $^3$He or $^4$He experiments such as
counter-flow, moving grids, colliding vortex rings \citep{SkrbekSreeni12,WalmsleyGolov08}, 
%towed grids, oscillating and vibrating grids in stationary superfluids, and collisions of 
%vortex rings attached to ions 
%\citep{Barenghietal06}, %\citep{Smithetal93}, \citep{Davisetal00,Bradleyetal05a} \citep{WalmsleyGolov08}. 
until recently very little has been known directly about the underlying microscopic 
interactions. Instead, the nature of the vortex interactions has been inferred from how
rapidly the tangle decays.

Theoretically, the observed decay has been linked to the conversion of the kinetic 
energy of the vortices into other forms of energy. This could be conversion into
the kinetic energy of the normal component in higher temperature experiments or into the 
interaction energy and quantum waves 
%can be explained if vortex interactions 
%are able to efficiently convert vortices' kinetic energy into either the interaction 
%energy of the system or quantum waves.  This loss of kinetic energy is allowed 
in low temperature quantum fluids, including Bose-Einstein condensates. 
%Even those without a quasi-classical normal fluid component capable of
%removing kinetic energy via viscous dissipation.
Despite this, most of our current theoretical insight into quantum vortex reconnection has 
been through Lagrangian, Biot-Savart simulations of isolated vortex filaments, a
dynamical system that does not include the terms for the interaction energy. Why?

Part of the reason is that the Lagrangian approach has experimental support, 
most recently by comparisons 
between experiments tracking quantum vortices with solid hydrogen particles
\citep{Bewleyetal_PNAS08,Paolettietal08} and the scaling in the 
filament calculation of 
initially anti-parallel vortices by \cite{deWaeleAarts94}. 
In both cases, the minimum separation distance between vortices, $\delta$, was interpreted
in terms of the dimensional analysis based upon the circulation $\Gamma$ of the vortices. 
That is, if $\dddt\delta\sim v\sim \Gamma/\delta$, then one would expect that
\EQL{eq:deltat} \delta(t)\sim (\Gamma |t_r-t|)^{1/2}\,. \EN
This will be called the {\it dimensional scaling}. 

Alternatively, one can simulate the underlying mean-field equations of quantum fluids and visualise vortex reconnection by following the low density isosurfaces that surround the zero density cores. The problem with this approach is that tracking the motion of the vortices within these isosurfaces is difficult, even for single interactions.

The aim of this paper is to begin to fill that gap using two innovations for solutions of the 
mean-field, hard-sphere Gross-Pitaevskii equations.  One innovation is an initial condition 
that suppresses fluctuations in the temporal scaling of separations and
the second is a method for identifying the position of the vortex cores. These innovations will be used
to determine scaling laws for two classes of initial configurations, orthogonal or anti-parallel 
vortices. 

The conclusion will be that the scaling laws for the minimum separation distance between the two vortices in the two configurations are distinctly different, even when the pairs are just several core radii apart. The anti-parallel case obeys the expectations from \eqref{eq:deltat}, but the orthogonal cases consistently obey a distinctly different type of scaling.  
The two sets of scaling laws will be associated with differences in the alignment of their 
respective Frenet-Serret coordinate frames, differences that form almost immediately

This paper is organised as follows. First, the equations and the initialisation
of the  model are introduced, followed by overviews of the anti-parallel and orthogonal
global evolution in three-dimensions. Next, the methods used to identify the trajectories of the vortices and the local properties of the Frenet-Serret frame, including curvature, are explained. The numerical results, arranged by the type of simulation, orthogonal and anti-parallel, are then described. The results include the time dependence of the separation of the vortices, the curvature along the vortices and the alignments in terms of the Frenet-Serret frames. Finally, the differences between the two classes of initial conditions are discussed and how these differences might affect the observed scaling laws for the approach and release of reconnecting vortices.
\section{Equations, numerics and initial condition}\label{sec:eqIC}

Following \cite{Berloff04}, the three-dimensional Gross-Piteavskii equations for the complex wave function or order-parameter $\psi$ are
\EQL{eq:GP} \frac{1}{i}\ppt{\psi} = E_v\nabla^2\psi + 
V(|\bx-\bx'|)\psi(1-|\psi|^2)\quad{\rm with}\quad E_v=0.5~~{\rm and}~~
V(|\bx-\bx'|)=0.5\delta(\bx-\bx')\,. \EN
These are the mean-field equations of a microscopic, quantum system with
$\hbar$ and $m$ non-dimensionalized to be 1, a chemical potential of $E_v=0.5$ 
and using the hard-sphere approximation for $V(|\bx-\bx'|)$.  They
are an example of a defocusing nonlinear Schr\"odinger equation. All calculations
in this paper will use (\ref{eq:GP}). 

% One can then write the following fluid equation:
These equations conserve mass:
\EQL{eq:GPmass} M = \int dV |\psi|^2 \EN 
and a Hamiltonian
\EQL{eq:GPHamiltonian} H = \frac{1}{2}\int dV \left[ \nabla \psi \cdot \nabla \psi^\dag 
+ 0.25(1-|\psi|^2)^2 \right] \EN
where $\psi^\dag$ is the complex conjugate of $\psi$. The local strength of the mass 
density, kinetic or gradient energy $K_{\psi}$ and the interaction energy $I$ are
\EQL{eq:GPenergies} \rho=|\psi|^2,\quad K_{\psi}=\frac{1}{2}|\nabla\psi|^2
\quad{\rm and}\quad I(\bx)=\frac{1}{4}(1-|\psi|^2)^2  \EN
Isosurfaces of $\rho$ are used in all the three-dimensional visualisations and 
$K_{\psi}$ are included in figures \ref{fig:antipstreamrhoT0T4}, \ref{fig:orthoazel}, 
\ref{fig:orthoNP}, \ref{fig:orthoT6-12} and \ref{fig:antipstreamrhoTtr}.

Gross-Pitaevskii calculations have previously identified the following
features of quantum vortex reconnection.  First, it has been demonstrated 
\citep{Leadbeateretal03,Berloff04,Kerr11} that the line length grows
just prior to reconnection, indicating a type of vortex stretching. 
Second, reconnection radiates energy, either as sound waves 
\citep{Berloff04,Leadbeateretal01,Leadbeateretal03}, non-linear refraction 
waves \citep{Berloff04,Zuccheretal2012} or strongly non-linear vortex rings 
\citep{Leadbeateretal03,Berloff04,Kerr11}.  About 10\% of the 
initial kinetic energy $K_{\psi}$ is lost by these means
during the initial reconnection \citep{Kerr11}.  With added terms representing
assumptions about the type of energy depletion at small-scales, 
these equations can also give us hints to why the vortex tangle decays
\citep{SasaKMLvovRTsubota11}.

\subsection{Quasi-classical approximations \label{sec:quasiclassic}}

But how can the continuum Gross-Pitaevskii equations provide us with details about the Lagrangian
dynamics and reconnections that underlie the vortex tangle of quantum fluids?

This can be done by writing the wave function as $\psi=\sqrt{\rho}e^{i\phi}$, 
where $\rho$ is the density and $\phi$ is the complex phase, then defining 
the phase velocity $\bv_\phi$ and quantised circulation $\Gamma$ around the line defects:
\EQL{eq:Madelung}
\bv_\phi=\nabla\phi={\rm Im}(\psi^\dag\nabla\psi)/\rho \qquad
\Gamma=\int \bv_\phi\cdot\bs=2\pi \EN
then identifing the $\rho\equiv0$ line defects as quantum vortices whose dimensionless quantised circulation is $\Gamma$. Even though these lines cannot represent a true vorticity field because the vorticity $\bomega=\nabla\times\bv_\phi=\nabla\times\nabla\cdot\phi\equiv0$. In this picture, vortex reconnection appears naturally as the instantaneous re-alignment of these lines and exchange of circulation when the line defects meet. If dimensions were added, the quantised circulation has the classical units of circulation: $\Gamma\sim \rho L^2/T$. Note these two differences with classical vortices governed by the Navier-Stokes equation: classical circulation is not quantised and viscous reconnection is never 100\%.

To extract the Lagrangian motion of quantum vortices from fields defined on three-dimensional meshes  these issues must be addressed:
\ITM\item As the state relaxes from its initial form, it should not be dominated by 
either interal waves (phonons) or strong fluctuations along the vortex trajectories
(Kelvin waves).
\item Second, a method is needed for identifying the direction and positions that the vortices follow as they pass through the three-dimensional mesh.
\ITN

Two innovations introduced in this paper, \eqref{eq:initBerloff} for the initial
density and \eqref{eq:pseudov} for tracking vortices,
resolve both problems and allow us to extract smooth motion for the vortices from the 
calculated solutions of the Gross-Pitaevskii equations on Eulerian meshes. 

To complete the discussion of the Gross-Pitaevskii equations \eqref{eq:GP}, the full analogy to classical hydrodynamic equations comes from inserting
$\psi=\sqrt{\rho}\exp(i\phi)$ into (\ref{eq:GP}) to get the standard equation 
for $\rho$ and a Bernoulli equation for $\phi$:
\EQL{eq:rhophi} \dppt\rho+\nabla\cdot(\rho\bv_\phi)=0\qquad
\ppt\phi+(\nabla\phi)^2=0.5(1-\rho)+\nabla^2\sqrt{\rho}/\sqrt{\rho}\,. \EN
The $\bv_\phi$ velocity equation can then be formed by taking 
the gradient of the $\phi$ equation.

As in \cite{Kerr11},
the numerics are a standard semi-implicit spectral algorithm where the
nonlinear terms are calculated in physical space, then transformed to
Fourier space to calculate the linear terms.  In Fourier space,
the linear part of the complex equation is solved through integrating factors with
the Fourier transformed nonlinearity added as a 3rd-order Runge-Kutta explicit 
forcing.  The domain is imposed by using no-stress cosine transforms in all
three directions. 
For all of the calculations the domain size 
is $L_x\times L_y\times L_z=(16\pi)^3\approx50^3$ or $(32\pi)^3$.
Both $128^3$ and $256^3$ grids were used, with the $256^3$ grid giving smoother
temporal evolution. Most of the analysis and graphics will use the $\delta_0=3$, $256^3$
calculation.  
\BFG
\bminic{0.48}
\vspace{-10mm}\includegraphics[scale=.175]{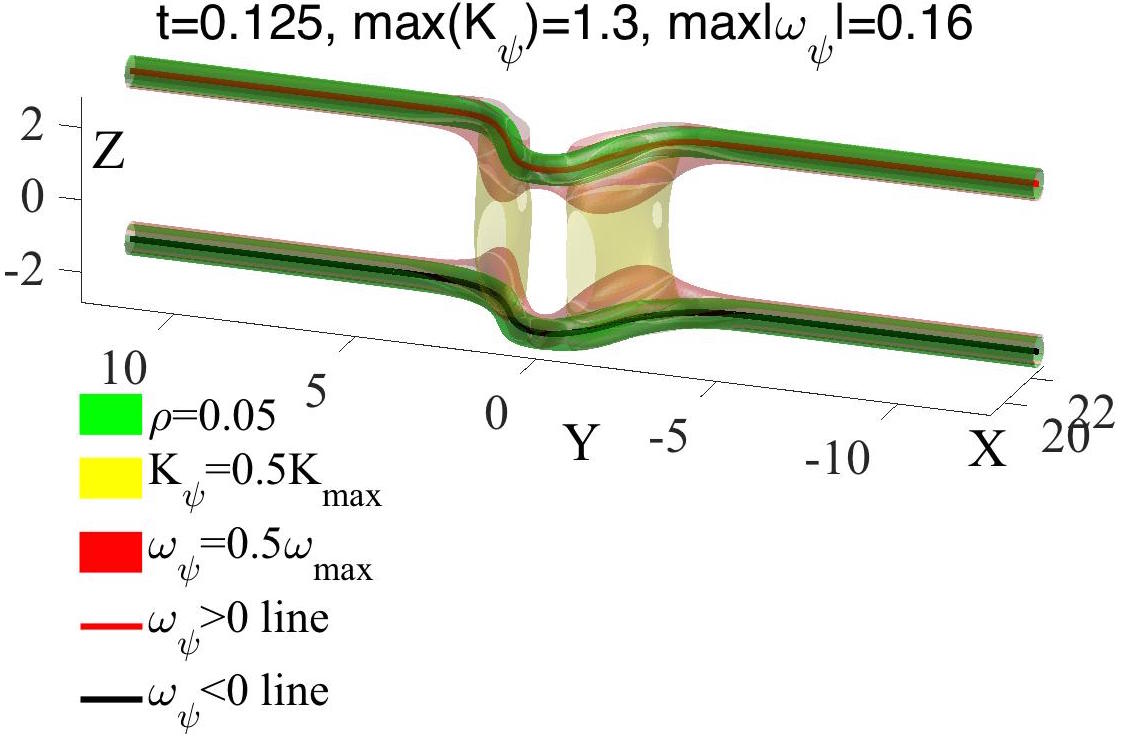}
\begin{picture}(0,0)
\put(0,130){(a)}
\end{picture}\emini~~~\bminic{0.48}
\includegraphics[scale=.175]{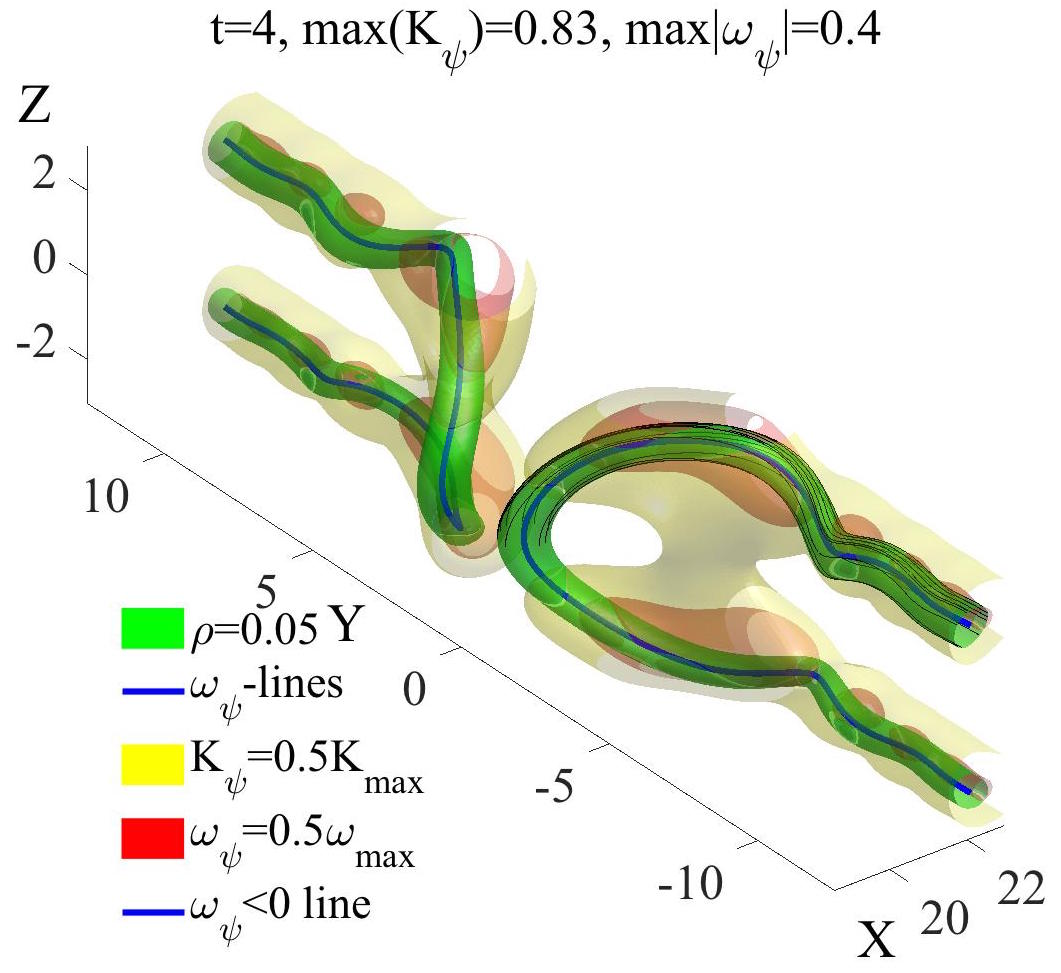}
\begin{picture}(0,0)
\put(0,168){(b)}
\end{picture}\emini
\caption{\label{fig:antipstreamrhoT0T4} Anti-parallel case: 
Density, kinetic energy $K_{\psi}=|\nabla\psi|^2$ and pseudo-vorticity 
$|\nabla\psi_r\times\nabla\psi_i|$ isosurfaces plus vortex lines at these times:
(a) $t=0.125$, with $\max(K_{\psi})=1.3$, 
and (b) $t=4$ with $\max(K_{\psi})=0.83$. $\max|\omega_\psi|=0.4$ for all the times.
The vortex lines show that the pseudo-vorticity method not only follows the
lines $\rho=0$ well, but the $t=4$ frame also shows that it can be used to 
follow isosurfaces of fixed $\rho$. The $t=0.125$ frame shows that initially the 
kinetic energy is largest between the elbows of the two vortices, with large
$\omega_\psi$ at the elbows.  The $t=4$ frame
shows the newly reconnected vortices as they are separating  with
undulations on the vortex lines that will form additional reconnections 
and vortex rings at later times.
}
\EFG

\BFG
\bminil{0.5}
\includegraphics[scale=.65]{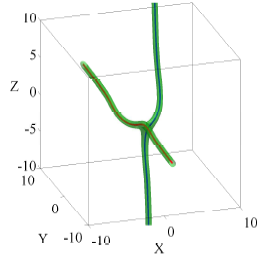}
\put(-140, 145){(a)}
\put(-95,145){t = 6}\\
\includegraphics[scale=.65]{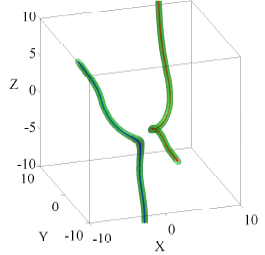}
\put(-140,135){(b)}
\put(-95,135){t = 14}
\caption{\label{fig:orthoazel}
Overview of the evolution of the $\delta_0=3$ $256^3$ orthogonal calculation from 
a three-dimensional perspective for two times: $t=6$ as reconnection is beginning and 
$t=14$ after it has ended, where $t_r\approx 8.9$. Green $\rho=0.05$ isosurfaces encase the
vortex cores, blue and red lines. The Frenet-Serret frames for $t=8.75$ are given 
in figure \ref{fig:orthoT6-12}. 
A twist in the post-reconnection $t=14$ right ($x>0$) vortex is visible.
}%\caption
~~~ \emini ~~~~~~\bminir{0.5}
\includegraphics[scale=.65]{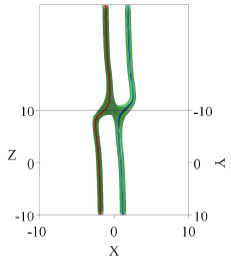}
\put(-140, 145){(a)}
\put(-95,145){t = 6}\\
\includegraphics[scale=.65]{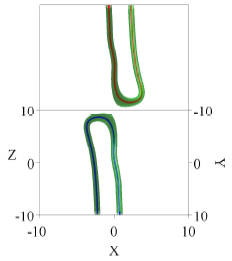}
\put(-140,135){(b)}
\put(-95,135){t = 14}
\caption{\label{fig:orthoNP} The same $\delta_0=3$ fields and times as in 
figure \ref{fig:orthoazel} from the Nazarenko perspective, which looks
down the 45$^\circ$ {\it propagation plane} in the $y-z$-plane in 
figure \ref{fig:yzplane}.  This perspective is useful
because the pre- and post-reconnection symmetries can be seen most clearly, which is
why it is the basis for comparison of reconnection angles in 
figure \ref{fig:NPplane} and section \ref{sec:orthoangles}. 
The arms are extending away 
from the {\it reconnection plane}, either towards or away from the viewer,
Also note that from this perspective, loops are clearly not forming.  }
\emini
\EFG

\subsection{Choice of initial configurations and profiles \label{sec:profiles}}

{\bf Configurations.} Two initial vortex configurations are used in this paper, 
anti-parallel vortices with a perturbation, and orthogonal vortices. 
Both configurations have been used many times 
for both classical (Navier-Stokes) and quantum fluids, including
the first calculations using the three-dimensional 
Gross-Pitaevskii equations \citep{KoplikLevine93}.  

The advantage of focusing on these configurations is that the interactions leading to
reconnection for most other configurations, for example colliding or initially 
linked, vortex rings, can be reduced to either anti-parallel or orthogonal 
dynamics, both of which can resolve the reconnection events in smaller global domains.
This is because the initial reconnection events require, 
effectively, only half of each ring. 

These two configurations also represent the two extremes for the initial chirality or 
linking number of the vortex lines in a quantum fluid.  In classical fluids the 
corresponding global property is the helicity, $h=\int (\nabla\times\bu)\cdot\bu\, dV$. 
When $h$ is large, it tends to suppress nonlinear interactions.  
%The global helicity $h$ is a quadratic invariant
%of the three-dimensional incompressible Euler equations in addition to the kinetic energy
Anti-parallel initial conditions have zero net helicity and
are sensitive to initial instabilities while orthogonal initial conditions have a
large helicity, so reconnection can be delayed \citep{BoratavPelzZ92}.

{\bf Density profile.} The density profiles for all the vortex cores in this paper
are determined by the following Pad\'e approximate:
\EQL{eq:initBerloff}\begin{array}{c}\medskip |\psi_{sb}|=\sqrt{\rho_{sb}}=
\dfrac{c_1r^2+c_2r^4}{1+d_1 r^2+d_2 r^4} \\
{\rm with}\quad c_1=0.3437,~~c_2=0.0286,~~d_1=0.3333,~~{\rm and}~~d_2=0.02864\,. 
\end{array}\EN

Note that $c_2\lesssim d_2$, which implies that as $r\rightarrow\infty$, the density
approaches the usual background of $\rho=1$ from below more slowly than the 
true Pad\'e of this order does. The true Pad\'e, derived by \cite{Berloff04}, has
$c_2=d_2=11/384\approx0.02864$. Therefore, the profile with $c_2\lesssim d_2$ is designated as 
$\psi_{sb}$ because it is a sub-\cite{Berloff04} profile. 
Furthermore, because the calculations are in finite domains, to ensure that 
the Neumann boundary conditions are met, a set of up to 24 mirror images of 
the vortices are multiplied together. This multiplication process takes the slight,
original $\rho_{sb}<1$ values as $r\rightarrow\infty$ and generates 
stronger differences. At the boundaries, this gives $(1-\rho)\approx0.02-0.03$. 

For all of the configurations discussed here, using this $|\psi|\lesssim1$ initial 
profile appears to be crucial in allowing us to obtain clear scaling laws for the 
pre- and post-reconnection separation of the vortices.  As discussed in
subsection \ref{sec:subBerloff}, further tests have confirmed that 
the temporal separations of all of the
true Pad\'e approximates have significant fluctuations.
%Using $c_2\lesssim d_2$ was discovered by chance and we do not have
%a clear explanation for why it works.
%%$c_0=0$, $c_1=p_1^2$, $d_1=0.75p_1^2/(1-p_1^2)\approxeq 0.3857$,
%%$ d_2=p_1^2[(p_1^2-0.25)/(1-p_1^2)]\approxeq 0.0461$.

\begin{table}
\centering
\begin{tabular}{c | c c c c c }
 \hline
type & mesh & ~~~~$\delta_0$ & $t_r$ & $K_{\psi}(t=0)$  & $I_0$ \\
 \hline
{\large$\perp$ $(16\pi)^3$} & $256^3$ & ~~~~2 & 2.7 & 0.0074 & 0.0023 \\
 \hline
{\large$\perp$ $(16\pi)^3$} & $128^3$ & ~~~~3 & 8.9 & 0.0074 & 0.0023 \\
 \hline
{\large$\perp$ $(16\pi)^3$} & $256^3$ & ~~~~3 & 8.9 & 0.0074 & 0.0023 \\
 \hline
{\large$\perp$ $(16\pi)^3$} & $256^3$ & ~~~~4 & 21 & 0.0074 & 0.0023 \\
 \hline
{\large$\perp$ $(32\pi)^2\times16\pi$} & $256^2\times128$ & ~~~~4 & 21 & 0.00178 & 0.00049 \\
 \hline
{\large$\perp$ $(32\pi)^3$} & $256^3$ & ~~~~5 & 40 & 0.0074 & 0.0023 \\
 \hline
{\large$\perp$ $(32\pi)^3$} & $256^3$ & ~~~~6 & 68 & 0.0074 & 0.0023 \\
 \hline
{\large$\|$ $(16\pi)^3$} & $256^3$ & ~~~~4 & 2.44 & 0.0056 & 0.0031 \\
\end{tabular}
%\begin{tabular}{ c | c | c | c | c |}
% \hline
%$\delta_0$ & $t_r$ & $K_0$  & $IE_0$ \\
% \hline
%2 & 2.7 & 0.07 & 0.1 \\
% \hline
%3 & 9 & 0.05 & 0.1 \\
% \hline
%4 & 21 & 0.07 & 0.1 \\
% \hline
%5 & 40 & 0.07 & 0.1 \\
% \hline
%6 & 68 & 0.07 & 0.1 \\
% \hline
%\end{tabular}
\caption{Cases: Type is orthogonal {\large$\perp$} or anti-parallel {\large$\|$}, initial 
separations, approximate reconnection times,
initial kinetic and interaction energies.}
\label{tab:ortho}
\end{table}

Once the best profile has been chosen, then one must choose the trajectories of the
interacting vortices.  The anti-parallel global states are shown first in
figure \ref{fig:antipstreamrhoT0T4} as they
illustrate the use of all of the three-dimensional diagostics. 

\subsection{Anti-parallel: Initial trajectory and global development.\label{sec:antiptraj}}

Based upon past experience with classical vortices and \cite{Kerr11}, the positions 
$\bs_{\pm}(y)$ of the two nearly anti-parallel $y$-vortices, with a perturbation 
in the direction of propagation $x$, were
\EQL{eq:newtraj} \bs_{\pm}(y) =
\left(\delta_x\frac{2}{\cosh\bigl((y/\delta_y)^{1.8}\bigr)}-x_c,y,\pm z_c\right)\,. \EN
The parameters used were $\delta_x=-1.6$, $\delta_y=1.25$, $x_c=21.04$ and $z_c=2.35$.
The power of 1.8 on the normalized position
was chosen to help localise the perturbation near the $y=0$ symmetry plane.  
The density profiles
were applied perpendicular to this trajectory, and not perpendicular to the $y$-axis.

As in \cite{Kerr11}, two of the Neumann boundaries act as symmetry planes
to increase the effective domain size. These planes are the $y=0$, $x-z$
{\it perturbation plane} and $z=0$, $x-y$ {\it dividing plane}.
Because the goal of this calculation was to focus upon the scaling around the first reconnection, the long domain used in \cite{Kerr11} to generate a chain of vortices is unnecessary and $L_y$ is less. In addition, based upon recent experience with Navier-Stokes reconnection \citep{Kerr13}, $L_z$ was increased to ensure that the evolving vortices do not see their mirror images across the upper $z$ Neumann boundary condition.

Figure \ref{fig:antipstreamrhoT0T4} shows the state at $t=0.125$, 
essentially the initial condition, and the state
at $t=4$, after the first reconnection event at $t_r\approx2.4$. Three isosurfaces are
given.  Low density isosurfaces ($\rho=0.05$), 
isosurfaces of the kinetic energy $K_{\psi}$ \eqref{eq:GPenergies}
and isosurfaces of $|\bomega_\psi|$ \eqref{eq:pseudov}, 
a pseudo-vorticity that is introduced in the next section. 
The vortex lines defined by $\bomega_\psi$ and Proposition 1 are shown using
thickened curves.
The structure at the time of reconnection is discussed in section \ref{sec:antipresults} 
using Figure \ref{fig:antipstreamrhoTtr} and how the flow would develop later 
has already been documented by \cite{Kerr11}, which shows several reconnections forming
a stack of vortex rings. 

Both $K_{\psi}$ and $\bomega_\psi$ are functions of the first-derivatives of the
wave function, but show different aspects of the flow. The $K_{\psi}$ 
isosurfaces show where the momentum is large. Initially, the momentum is  dominated by 
forward motion between the perturbations, as shown for $t=0.125$.
Post-reconnection, at $t=4$, the $K_{\psi}$ surfaces show that the 
primary motion is around the vortices.  $|\bomega_\psi|$ is large where the 
vortex cores bulge and have the greatest curvature. 

\subsection{Orthogonal: Initial separations and global development.
\label{sec:ortho}}

To place the orthogonal vortices one only needs to choose one line parallel to the 
$y$ axis and another line parallel to the $z$ axis through two points in $x$ on either 
side of $x=y=z=0$.  The five separations and other details of the simulations 
are given in Table \ref{tab:ortho}.
Because all of the orthogonal cases with $\delta_0\geq2$ behave qualitatively in the same 
manner, all of the orthogonal three-dimensional images will be taken from the $256^3$ 
$\delta_0=3$ calculation, whose estimated reconnection is at time $t_r=8.9$.  

The first two sets of three-dimensional isosurfaces and vortex lines in figures 
\ref{fig:orthoazel} and \ref{fig:orthoNP} are used to show the global evolution
of the vortices through reconnection, with  figure \ref{fig:orthoazel}
providing a true three-dimensional perspective of the orthogonality and
figure \ref{fig:orthoNP}, defined by figure \ref{fig:NPplane}, providing a 
perspective down the $y=z$ axis, which is also used for the determination of 
three-dimensional angles in subsection \ref{sec:orthocurve}. 
The two times chosen for each are $t=6$, pre-reconnection, and $t=19$, post-reconnection.

The two isosurfaces are for a low density of $\rho=0.05$ and kinetic energy of
$K_{\psi}=0.5$, where $\max(K_{\psi})=0.58$. 
There are two pseudo-vortex lines 
\eqref{eq:pseudov} in each frame, one that originates on the $y=0$ plane 
and the other on the $z=0$ plane.
The $t=8.5$ frame also shows some additional 
orientation vectors that will be discussed in section \ref{sec:orthoresults}.  

{\bf Qualitative features are:}
\ITM\item The initially orthogonal vortices are attracted towards each other 
at their points of closest approach, asymmetrically bending out towards each other. 
\item During this stage there is a some loss of the kinetic energy between $t=8$ and 
$t=11$, $\Delta_t K_{\psi}<10\% K_{\psi}(t=0)$. This
is converted into interaction energy $I$ (\ref{eq:GPenergies}). 
There are no further noticeable changes in $K_{\psi}$ for $t>12$. 
\item After reconnection, from one perspective there is
a slight twist on one vortex, but it is not twisted enough for the vortices to
loop back upon themselves and reconnect again. Instead, the two new vortices pull back 
from one another, as shown by the $t=19$ frame in
figure \ref{fig:orthoazel}.  
Consistent with experimental observations of vortex interactions 
using solid hydrogen particles \citep{Bewleyetal_PNAS08,PhysTodayJul10}
in the sense that post-reconnection filaments simply pull back from one 
another and do not loop.  
\item Two sketches of the alignments and an addition isosurface perspective are
used in subsection \ref{sec:orthogeo} to illustrate this evolution further.
\ITN
\section{Approach and separation of vortex lines: Methodology}\label{sec:inout}

The primary result in this paper will be the differences in the temporal scaling 
of the pre- and post-reconnection separation of the vortices for the two
configurations.  The secondary results are clues for the origin of these 
differences in the evolution of the local curvature and Frenet-Serret frames 
on the vortices for the two configurations.
To achieve this, one needs an initial condition for which the evolution of
the vortices is smooth and a means to follow that evolution.  The key ingredient
of the initial condition is provided by the choice of coefficients in
\eqref{eq:initBerloff}.  This section will show what is needed to accurately
track the vortices.  Two methods for detecting the vortices have been used.  

\subsection{Detecting lines by finding $\rho=0$ mesh cells.}\label{sec:meshcells}

The first approach is to estimate the locations of the $\rho\equiv0$ quantum 
vortex lines by extracting the positions of the vertices of a $\rho\gtrsim0$ 
isosurface mesh, determined by Matlab, then average these positions.
For this method to work, the density $\rho$ has to be small enough so that there
are only 3-6 points clustered in a plane perpendicular to the vortex lines.
The method begins to fail around reconnection points because there is an extensive
$\rho\approx0$ zone as the cores of the vortices start to overlap, resulting in
the isosurface points that are too far apart to make reliable estimates 
for the positions of the $\rho\equiv0$ cores.  
Due to these problems, this approach is used only 
for providing the seeds for our preferred approach and its validation.  

\subsection{$\bomega_\psi=\nabla\rho\times\nabla\phi$ pseudo-vorticity method\label{sec:pseudo}}

The second approach begins by recognising that the line of zero density should be 
perpendicular to the gradients of the real and imaginary parts of the wave function.
Therefore it is useful to define the following {\it pseudo-vorticity}:
\EQL{eq:pseudov} \bomega_\psi=0.5\nabla\psi_r\times\nabla\psi_i \EN

The inspiration for this approach comes from how to write the vorticity in terms
of a cross production of the scalars in Clebsch pairs, which is an alternative 
approach to representing  the incompressible Euler equations.

\begin{proposition}\label{th:omegarho} At points with $\rho=0$, the direction of the quantum vortex line is
defined by the direction of $\bomega_\psi=0.5\nabla\psi_r\times\nabla\psi_i$.
\end{proposition}
\begin{proof}
A quantum vortex line is defined by $\rho \equiv 0$, which because $\rho=\psi_r^2+\psi_i^2$ implies that $\psi_r=\psi_i=0$ on this vortex line.

Then define $\hat{\omega}_\rho$ the direction vector of the line at any arbitary point on the vortex line. By the definition of the line, the values of $\psi_r$ and $\psi_i$ in this direction must not change and thus we know that $\nabla\psi$ must satisfy
$$\hat{\omega}_\rho\cdot\nabla\psi_r=\hat{\omega}_\rho\cdot\nabla\psi_i=0 $$ 
at these points. This is only possible if 
$\hat{\bomega}_\rho=(\nabla\psi_r\times\nabla\psi_i)/
|\nabla\psi_r\times\nabla\psi_i|$.
\end{proof}

%%%%%%%%%%%%%%%%% 
% \ITM\item {\it Proof:} Begin with the definition of a quantum vortex line $\rho=\psi_r^2+\psi_i^2$
% as being a line along which $\rho\equiv0$.
% \item Because $\rho=\psi_r^2+\psi_i^2$, this implies that in addition 
% $\psi_r=\psi_i=0$ on this line.
% \item From this it follows that if the direction vector $\hat{\omega}_\rho$ of that 
% line is known at some point, then 
% $\hat{\omega}_\rho\cdot\nabla\psi_r=\hat{\omega}_\rho\cdot\nabla\psi_i=0$ at
% that point.
% \item Therefore $\hat{\omega}_\rho$ is perpendicular to both $\nabla\psi_r$ and
% $\nabla\psi_i$ at that point.
% \item This is only possible if 
% $\hat{\bomega}_\rho=(\nabla\psi_r\times\nabla\psi_i)/
% |\nabla\psi_r\times\nabla\psi_i|$ $\square$.
% \ITN
%%%%%%%%%%%%%%%%%%
By itself, this proposition does not tell us where the vortex lines lie because
one still needs a method for identifying a point on the line.  
To find starting points for a streamline function, the first method is used to
identify points on the boundaries where $\rho\approx0$.  

Potentially there could 
have been difficulties near the time and position of reconnection because both 
$\rho\approx0$ and $\nabla\psi_{r,i}$ are small, perhaps too small for the 
identifying the positions of neigbouring lines with $\rho\approx0$.  
In practice, this has not been a problem.

Once the lines have been found, the derivatives along their trajectories
of their three-dimensional positions can be determined, and from those
derivatives the local curvature, Frenet-Serret coordinate frames and possibly
the local motion of the lines can be found.
Properties that could be compared to the predictions of vortex filament models.  

To analyse these properties, the following alternative definition of the 
pseudo-vorticity is useful.

\begin{corollary}\label{th:omegarhophi} $\hat{\bomega}_\rho=\hat{\bomega}_\psi$ where
$\hat{\bomega}_\rho=\nabla\rho\times\nabla\phi/|\nabla\rho\times\nabla\phi|$ 
\end{corollary}
\begin{proof}
Start with $\psi_r=\sqrt{\rho}\cos\phi$ and 
$\psi_i=\sqrt{\rho}\sin\phi$. 
\ITM\item[] Expand: $0.5\nabla\psi_r\times\nabla\psi_i=
[\nabla\sqrt{\rho}\cos\phi-\sqrt{\rho}\sin\phi\nabla\phi]\times[\nabla\sqrt{\rho}\sin\phi+\sqrt{\rho}\cos\phi\nabla\phi]$. 
\item[] Remove all $\psi_r$ and $\psi_i$ terms sharing the same gradient to reduce 
this to 
$$2(\nabla\sqrt{\rho}\times\nabla\phi\sqrt{\rho}\cos^2\phi-\nabla\phi\times\nabla\sqrt{\rho}\sqrt{\rho}\sin^2\phi)$$ 
\item[] Finally, use $\nabla\sqrt{\rho}=\nabla\rho/(2\sqrt{\rho})$ to
get $\bomega_{\psi}=\nabla\rho\times\nabla\phi$. 
\ITN
\end{proof}

Do these lines follow the cores of $\rho\equiv0$?  One test is to interpolate
the densities from the Cartesian mesh to the vortex lines. The result is that
these densities are very small, but not exactly zero.  Another test is 
simultaneously plot the pseudo-vorticity lines along with very low isosurfaces of
density, examples of which is given in figure \ref{fig:antipstreamrhoT0T4}
and figure \ref{fig:orthoazel}.
The centres of the isosurfaces and the lines are almost indistinguishable.

Using the next proposition, the motion of the $\rho=0$ lines given by the
time derivative of $\psi$ can be written exactly using just the gradients and 
Laplacians of the wavefunction $\psi$.
This will will be used in a later paper.

\begin{proposition}\label{th:rho0v}
 The motion of the vortex line is given by the coupled set of equations 
 
 $$\begin{array}{rcl}\medskip (\nabla\psi_r\times\nabla\psi_i)\cdot\dddt\bx(s,t) 
& = & 0\\ \medskip
-\nabla\psi_r\cdot\dddt\bx(s,t) & = & -0.5\Delta\psi_i \\
-\nabla\psi_i\cdot\dddt\bx(s,t) & = & 0.5\Delta\psi_r \end{array}$$.

The solution of which is
\EQL{eq:rho0v} \dddt\bx(s,t)=-0.5\frac{\Delta\psi_i(\bomega_\psi\times\nabla\psi_i)
+\Delta\psi_r(\bomega_\psi\times\nabla\psi_i)}{\bomega_\psi^2} \EN

where pseudovorticity $\bomega_\psi:=\nabla\psi_r\times\nabla\psi_i$
\end{proposition}

\begin{proof}
We already know that the trajectory of the vortex lines is defined by the
pseudovorticity $\bomega_\psi=\nabla\psi_r\times\nabla\psi_i$ from proposition above. 

Since the
density remains zero along this line, the motion we are interested in is 
perpendicular to this direction.

On the $\rho=\psi_r^2+\psi_i^2\equiv0$ lines the time derivatives of 
$\psi_{r,i}$ are:
$$\ppt\psi_r=-0.5\Delta\psi_i $$
$$\ppt\psi_i=0.5\Delta\psi_r $$

Next we can Taylor expand to first order $\psi_{r,i}$ about the parameterised curve $\bx(s,t)$.
$$\psi_r=(\nabla\psi_r)(\bx-\bx(s,t))\quad{\rm and}\quad 
\psi_i=(\nabla\psi_i)(\bx-\bx(s,t)) $$
and their time-derivatives again to first order are
\begin{align*}
 \ppt\psi_r&=(\nabla\ppt\psi_r)(\bx-\bx(s,t))-\nabla\psi_r\ddt\bx(s,t)
\approx -\nabla\psi_r\ddt\bx(s,t)\\ 
\ppt\psi_i&=(\nabla\ppt\psi_i)(\bx-\bx(s,t))-\nabla\psi_i\ddt\bx(s,t)
\approx -\nabla\psi_i\ddt\bx(s,t)
\end{align*}

By adding that the motion will be perpendicular to the vortex (i.e. the
pseudovorticity $\nabla\psi_r\times\nabla\psi_i$) to the two
time derivative equations, one gets the required three coupled equations.

\end{proof}

\subsection{Curvature obtained from the $\omega_\psi$ lines}\label{sec:curve}

The curvature of the lines identified by the pseudo-vorticity algorithm will be found
by applying the Frenet-Serret relations to derivatives of the trajectories $\br(s)$ of the
vortex lines. 

\paragraph{Definition 3.1} {\it The Frenet-Serret frame for any smooth curve 
$\br(s):[0,1]\rightarrow\bbR^3$ has an orthonormal triple of unit vectors 
$(\bT,\bN,\bB)$ at each point $\br(s)$ where $\bT(s)$ is the  tangent, $\bN(s)$ 
is the normal and $\bB(s)$ is the binormal. The following relations between $(\bT,\bN,\bB)$ 
define the curvature $\kappa$ and torsion $\tau$. 
\begin{subequations}\label{eq:FrenetS}
\begin{align}\medskip \bT(s) & =\p_s\br(s) \label{tangent} \\
\medskip \p_s\bT & =\kappa \bN \label{eq:normal} \\
\medskip \p_s\bN & =\tau\bB-\kappa\bT \label{eq:gradN} \\
\p_s\bB & =-\tau\bN \label{eq:gradB}
\end{align}\end{subequations}
}

The numerical algorithm for calculating the curvature and normal uses the 
function {\it gradient} in Matlab twice. 
That is, first $\br_{,s}$ and then $\br_{,ss}$ are generated. Next, normalising $\br_{,s}$ 
gives the tangent vector $\bT$, the direction vector between points on the vortex lines.
Finally, the derivative of $\bT$ gives us both the curvature, $\kappa=|\p_s\bT|$
and the normal $\bN=\p_s\bT/\kappa$. In practice it is better to calculate the 
curvature using:
\EQL{eq:curvature} \kappa = |\br_{,s}\times\br_{,ss}|/|\br_{,s}|^3 \EN
% mag(cross(dr,ddr),1)./((mag(dr,1)).^3);
%\input{Ortho07sep14}
%%%%%%% RMK %%%%%%%%%%%%%%%%%%%%%%%%%%%%%%%%
\section{Orthogonal reconnection: New scaling laws and their geometry \label{sec:orthoresults}} 

The goals of this section are to  to apply the pseudo-vorticity algorithm \eqref{eq:pseudov} to the evolution of the initially orthogonal vortex lines and use these positions to demonstrate that the separation scaling laws for the originally orthogal vortices deviate strongly from the mean-field prediction for all initial separations and for all times.

The major points to be demonstrated for the orthogonal calculations are:
\ITM\item For strictly orthogonal initial vortices, there is just one reconnection and 
loops do not form out of the post-reconnection vortices in figure \ref{fig:orthoazel}.
\item The sub-Berloff profiles are crucial for obtaining temporal evolution that is smooth 
enough to allow clear scaling laws for the pre- and post-reconnection separtations
to be determined \citep{Rorai12}.
\item The separation scaling laws before and after reconnection are the same for each 
case, with these surprising results. Before reconnection 
$\delta_{in}\sim (t_r-t)^{1/3}$, which is slower than the dimensional scaling 
\eqref{eq:deltat}. And after reconnection $\delta_{in}\sim (t_r-t)^{2/3}$, faster than 
the dimensional scaling \eqref{eq:deltat}. 
\item This non-dimensional scaling arises as soon as the vorticity tangent vectors  
at their closest points are anti-parallel and the alignment of the averaged 
Frenet-Serret frames at these points with respect to the separation vector are
respectively orthogonal, parallel and orthogonal for the averaged tangent, curvature and 
bi-normal.
\item Reconnection occurs in the {\it reconnection or osculating plane} defined by 
the vorticity and curvature vectors at $t=t_r$ and is for all times approximately 
the plane defined by the average vorticity and curvature vectors of the two vortices 
at the points of closest approach. 
\item Angles taken between the reconnection event and the larger scale structure 
are convex, not concave or acute, which could be the source of the non-dimensional
separation scaling laws.
\ITN
%%% 
\subsection{Approach and separation}

The steps used to determine the separation scaling laws are these:
\ITM\item First, identify the trajectories of the vortex lines with the 
pseudo-vorticity plus Matlab streamline algorithm. At any given time, 
both before and after reconnection:
\ITM\item The vortex originating on the $y=0$ plane will be the $y$-vortex. 
\item The vortex originating on the $z=0$ plane will be the $z$-vortex.  \ITN
\item Identify the points, $\bx_y$ and $\bx_z$, of minimum separation between 
the two vortex lines, defined as $\delta_{yz}(t)=|\bx_y-\bx_z|$
\ITM\item and identify approximate reconnection times $\tilde{t}_r(\delta_0)$ 
when $\delta_{yz}(t)$ was minimal.
\item This generates $\delta_{yz}$ versus $t-t_r$ curves such as those in the inset of 
figure\ref{fig:orthoseparations}. \ITN
\item Once it was clear that neither the incoming nor outgoing separations obeyed 
the dimensional expectation \eqref{eq:deltat},
several alternative scaling laws were applied to the separations. Only the
1/3 incoming power law and outgoing 2/3 law working well for every case. 
\item By using these scaling laws to make the approach and separation linear, refined estimates 
of $\tilde{t}_r(\delta_0)$ can then be made. 
\ITM\item That is, cube the $\delta$ separations for $t<\tilde{t}_r(\delta_0)$. 
\item And take the 3/2 power of $\delta$ for $t>\tilde{t}_r(\delta_0)$. 
\item Then extrapolate these linear fits to the times when $\delta=0$. 
\item For all the initial $\delta_0$, the $t<t_r$ and $t>t_r$ estimates of $t_r$ were nearly identical.
\ITN
\item The combined results give the fit $t_r=(0.67\delta_0+0.064)^3$ and
are shown in figure \ref{fig:delta_trecon}. 
\item Using these $t_r(\delta_0)$, figure \ref{fig:orthoseparations} compares 
the scaled pre- and post-reconnection separations $\delta(t)$ for all 5 cases to demonstrate 
as in \cite{Rorai12} that:
\ITM
\item $\delta_{in}\sim |t_r-t|^{1/3}$ for $t<t_r(\delta_0)$ and 
$\delta_{out}\sim |t-t_r|^{2/3}$ for $t>t_r(\delta_0)$.
\item Note the inset which uses the dimensional scaling $\delta_{yz}^2$ versus time 
to illustrate the differences between the new scaling laws and the dimensional prediction.
\ITN
\ITN

%%% 
\BFG\bminil{0.5}

\includegraphics[scale=.31]{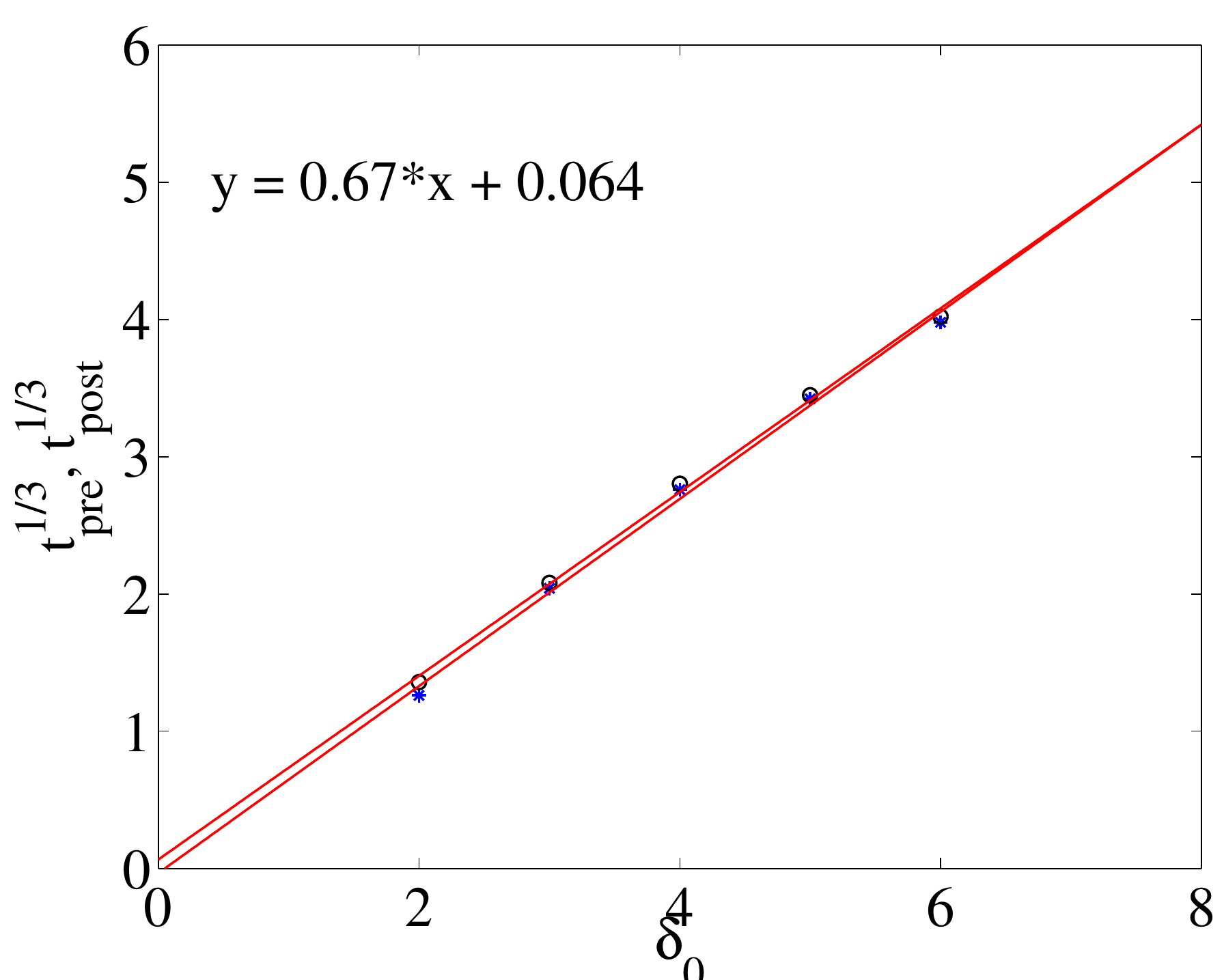} 
\emini~~~~\bminir{0.42} \caption{\label{fig:delta_trecon} 
Reconnection times as a function of the initial separation.  Reconnection times 
are estimated by using times immediately before and after reconnection plus 
the empirical 1/3 and 2/3 scaling laws.}
\emini
\EFG
%%% FIG 4.1
\BFG
\bminil{0.5}
\includegraphics[scale=.31]{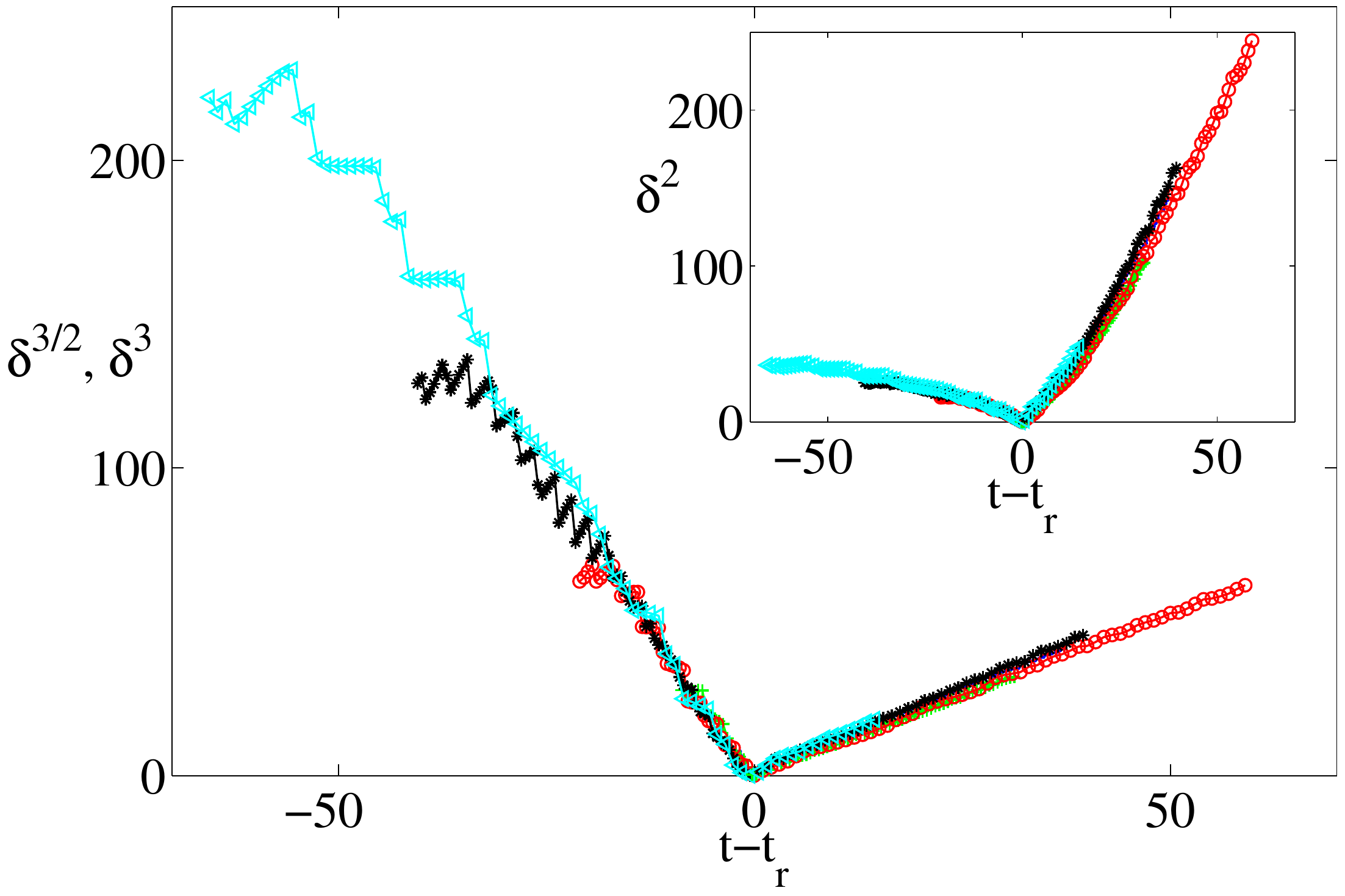} 
\emini~~~~\bminir{0.42} 
\caption{\label{fig:orthoseparations}
Pre- and post-reconnection separations for four cases: $\delta_0= 2, 3,4,5,6$
for calculations with $256^3$ points.  $\textcolor{blue}{\bf \bullet}$: $\delta_0=2$, $\textcolor{SeaGreen}{\bf +}$: $\delta_0=3$,
$\textcolor{red}{\bigcirc}$: $\delta_0=4$, ${\bf \divideontimes}$: $\delta_0=5$,
$\textcolor{cyan}{\triangle}$: $\delta_0=6$.
Pre-reconnection distances are raised to the power 3, while post-reconnection distances are raised to the power 3/2. This scaling is visibly better than the dimensional prediction ($\delta^2$ versus time to reconnection) shown in the inset. 
The scaled separations for $t<t_r$ fluctuate more strongly than the $t>t_r$ scaled 
separations.}
\emini\EFG
%%% FIG 4.2

\subsection{Sub-Berloff profile. \label{sec:subBerloff}} 
Why was it necessary to use the sub-Berloff profile? 
That is, could a different profile give similar separation collapse 
for all the $\delta_0$ cases and obtain clear scaling laws, as in 
figure \ref{fig:orthoseparations}? To show the benefits of the sub-Berloff
profile, tests were done using all of the known Pad\'e approximate profiles 
of steady-state two-dimensional quantum vortices in an infinite domain, including
tests with and without adding the mirror images.
This included, the true Berloff profile, that is (\ref{eq:initBerloff}) with 
$c_2=d_2=0.02864$, the low-order $2\times2$ Pad\'e approximate of the \cite{Fetter69},
(\ref{eq:initBerloff}) with $d_1=1$ and $c_2=d_2=0$, and a
%$c_0=0$, $c_1=p_1^2$, $d_1=0.75p_1^2/(1-p_1^2)\approxeq 0.3857$,
%$ d_2=p_1^2[(p_1^2-0.25)/(1-p_1^2)]\approxeq 0.0461$.
variety of $3\times3$ Pad\'e approximates from \cite{Berloff04} and most recently 
\cite{Roraietal13}, solutions that are very close to the ideal diffusive solution.  
All gave roughly the same oscillations in the approach and separation curves as in
\cite{Zuccheretal2012} and only a hint of the clear scaling laws in figure 
\ref{fig:orthoseparations}.  Only the sub-Berloff profile with at least some of the 
mirror images worked.  Further work will be needed to identify why instabilities generated
on the vortex lines are either suppressed by the sub-Berloff profile, or absorbed by it.

%%%%%
%%%%%
\subsection{Evolution of the orthogonal geometry during reconnection\label{sec:orthogeo}}
\BFG
\bminil{0.5}
\includegraphics[scale=.20]{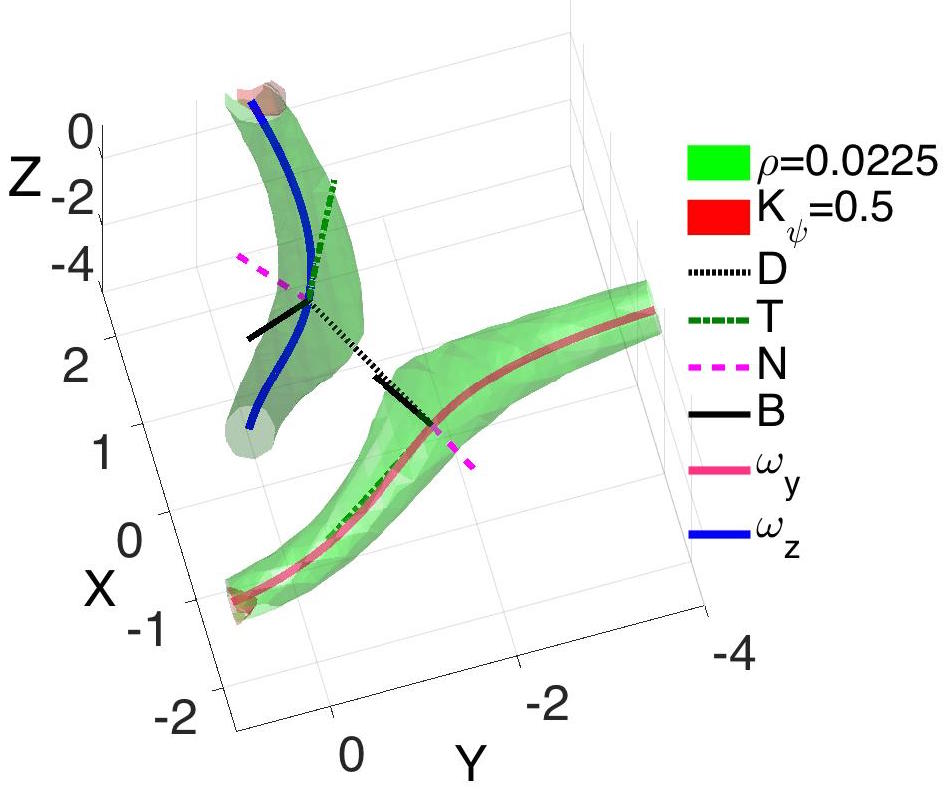}
\begin{picture}(0,0)
\put(-160,156){$t=6$}
\put(-180,156){\large a)}
\end{picture}\\
\includegraphics[scale=.20]{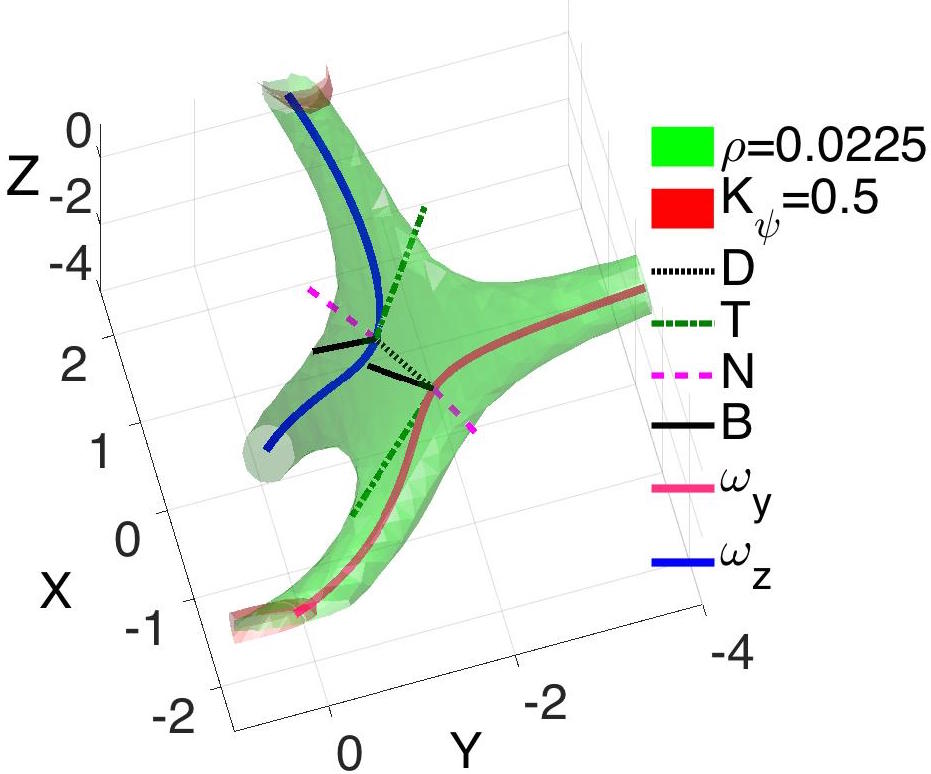}
\begin{picture}(0,0)
\put(-160,150){$t=8.75\approx t_r$}
\put(-180,150){\large b)}
\end{picture}\\
\includegraphics[scale=.20]{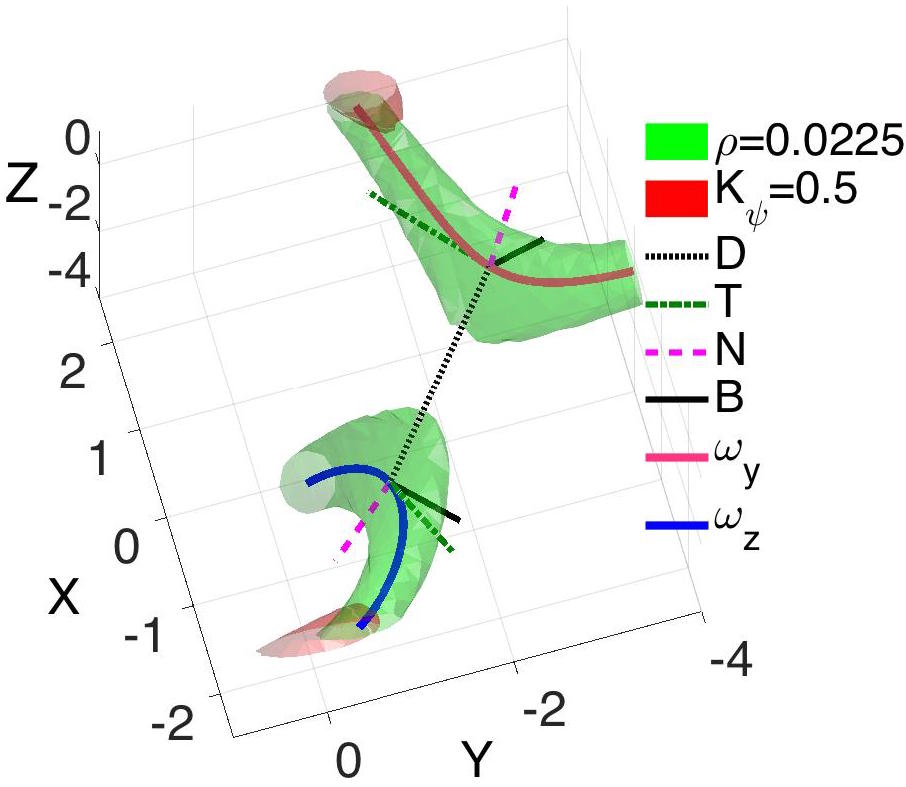}
\begin{picture}(0,0)
\put(-160,150){$t = 12$}
\put(-180,150){\large c)}
\end{picture}
\caption{\label{fig:orthoT6-12} % orthoT6-12} 
Isosurfaces, quantum vortex lines and orientation vectors in three-dimensions 
for the $\delta_0=3$, $256^3$ calculation at $t=6$ (\ref{fig:orthoT6-12}a), 
$t=8.75$  (\ref{fig:orthoT6-12}b) and $t=12$  (\ref{fig:orthoT6-12}c) where
$t_r\approx8.9$. Lines and vectors at the closest points on each vortex are the
separation $D$ and Frenet-Serret vectors \eqref{eq:FrenetS}: $\bT$, $\bN$ and $\bB$.
}
\emini~~~~\bminic{0.5}\vspace{5mm}
\includegraphics[scale=.180]{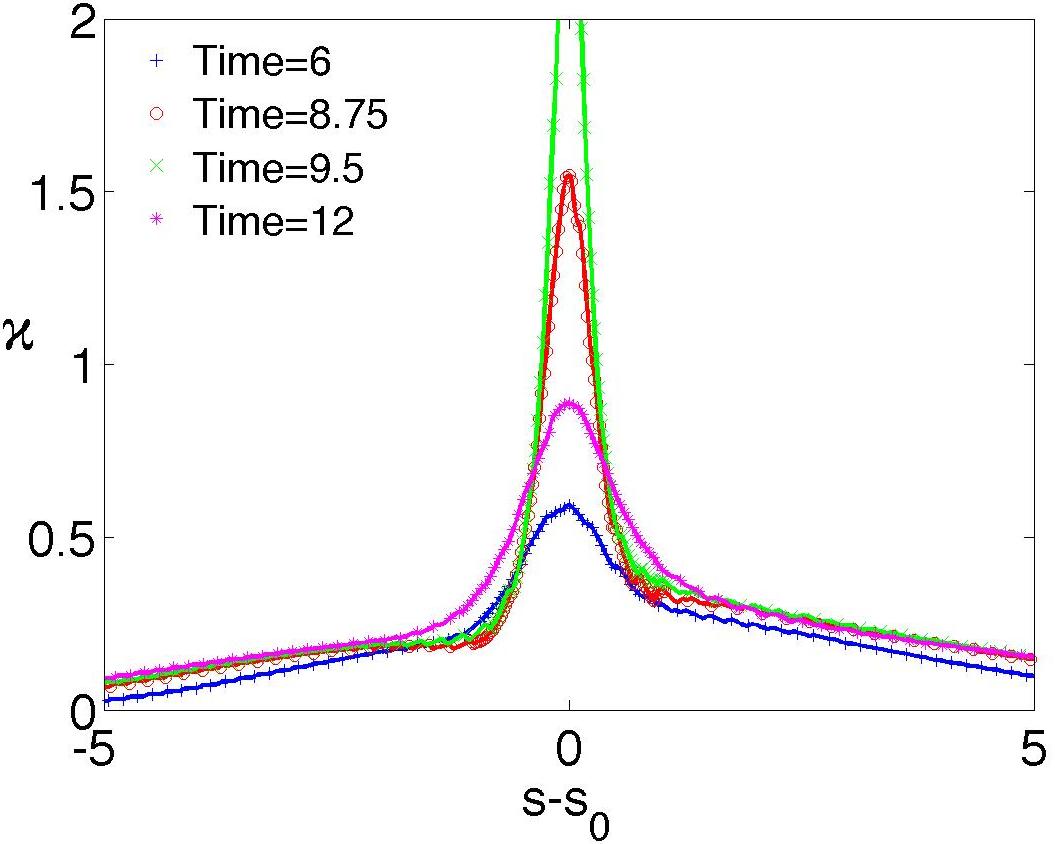}
\begin{picture}(0,0)
\put(30,100){\large a)}
\end{picture}\\ ~~~~ \\ ~~~~ \\ ~~~~ \\
\includegraphics[scale=.180]{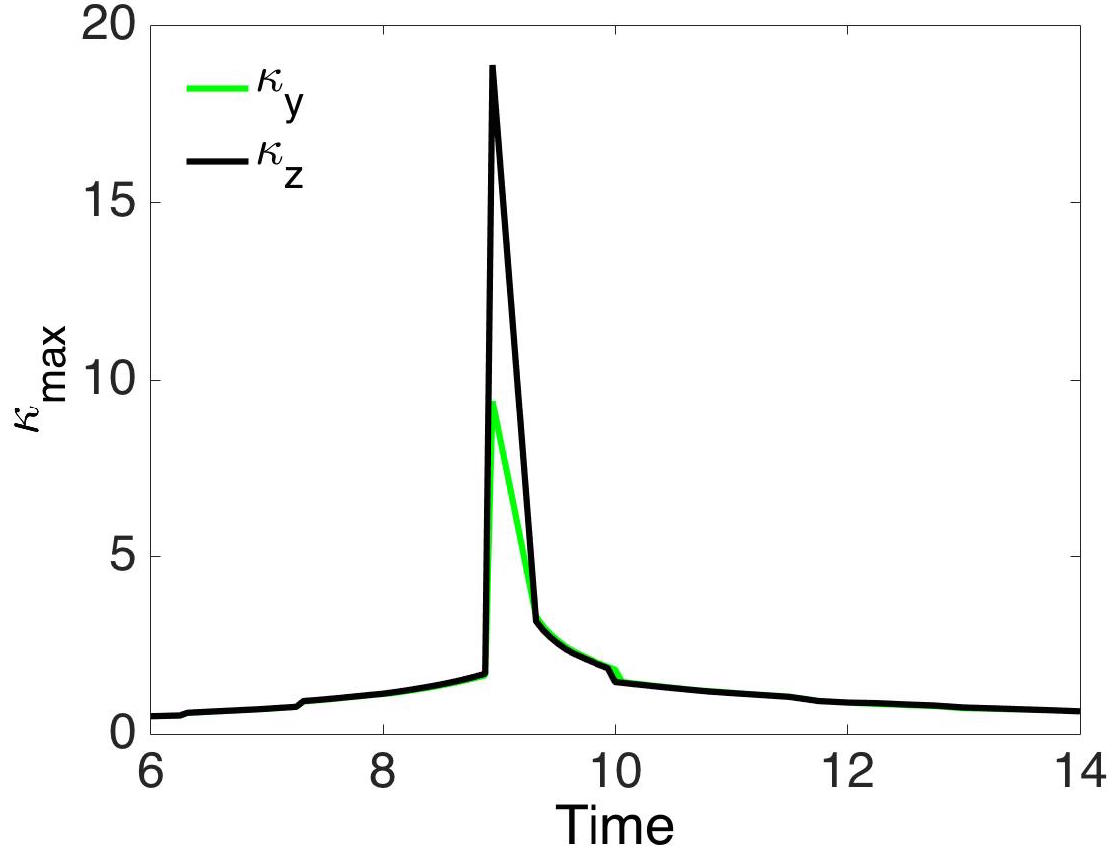}
\begin{picture}(0,0)
\put(40,100){\large b)}
\end{picture}\\ ~~~~ \\ ~~~~ \\ ~~~~ \\
\caption{\label{fig:orthocurv} 
Curvature of the vortex lines is given only for the $\delta_0=3$ $256^3$ calculation
with $t_r\approx 8.9$ as the curves about $t_r(\delta_0)$ are similar for all $\delta_0$.
(a) Curvatures against arclength $s$ at four times: $t=6$, 8.75, 9.5 and 12.
The profiles for the other line are similar, with their maximum peaks at 
$t=9.5$ both off the plotted scale. 
(b) Time evolution of the maximum curvatures of the two lines. For each line there is
a sudden jump to the large post-reconnection maximum at $t=t_r$.
}
\emini
\EFG

%%% FIG 4.3

The three-dimensional evolution of the vortices is illustrated in figures 
\ref{fig:orthoazel}, \ref{fig:orthoNP} and \ref{fig:orthoT6-12} and the alignments at or 
near reconnection are illustrated with two sketches taken from different perspectives in
figures \ref{fig:yzplane} and \ref{fig:NPplane}. 
The purpose of the sketches is to emphasize the strong qualitative differences between 
the orthogonal reconnection's skew-symmetric alignment and the anti-parallel case with 
its planar symmetries. Quantitative alignments are then given in figures \ref{fig:TNBD}, 
\ref{fig:D-TNB} and \ref{fig:angles2D}.
This initial discussion is
divided into three parts. First, the choice of three-dimensional images. Second,
the role of the sketches. Third, how to use your fingers to put the pieces together into
a mental picture.

\AD{The evolution of the three-dimensional structure is illustrated with the help of two
sketches in figures \ref{fig:yzplane} and \ref{fig:NPplane} and the perspectives 
in figures \ref{fig:orthoazel}, \ref{fig:orthoNP} and \ref{fig:orthoT6-12}.
The goal is to emphasize the differences with the anti-parallel case, give us some clues for why the separation laws are different and provide the reader with a mental picture of the evolution of a system that is skew-symmetric and harder to visualise than the anti-parallel case with its planar symmetries.}

{\bf Choice of 3D images.}
Section \ref{sec:profiles} uses figures \ref{fig:orthoazel} and \ref{fig:orthoNP}
to illustrate the global changes in structure from two perspectives.  One is a general 
perspective and the other is the Nazarenko perspective that shows the 
symmetries.  Each figure has a $t=6$ frame, long before the reconnection at $t_r=8.9$,
and a frame at $t=14$, long after reconnection.

Figure \ref{fig:orthoT6-12} focuses upon the reconnection zone using three times:
$t=6$ is at the beginning of reconnection, $t=8.75$ is just before the reconnection time 
of $t_r\approx8.9$ and $t=12$ shows the end of reconnection.  Over this period 
the $\rho=0.05$ isosurfaces change slowly while 
the $\rho\equiv0$ pseudo-vorticity lines within them move rapidly towards one another. 
The Frenet-Serret frames around the points of closest 
approach are discussed in subsections \ref{sec:orthocurve} and \ref{sec:orthoorient}. 

\BFG
\bminil{0.6}
\includegraphics[scale=.22]{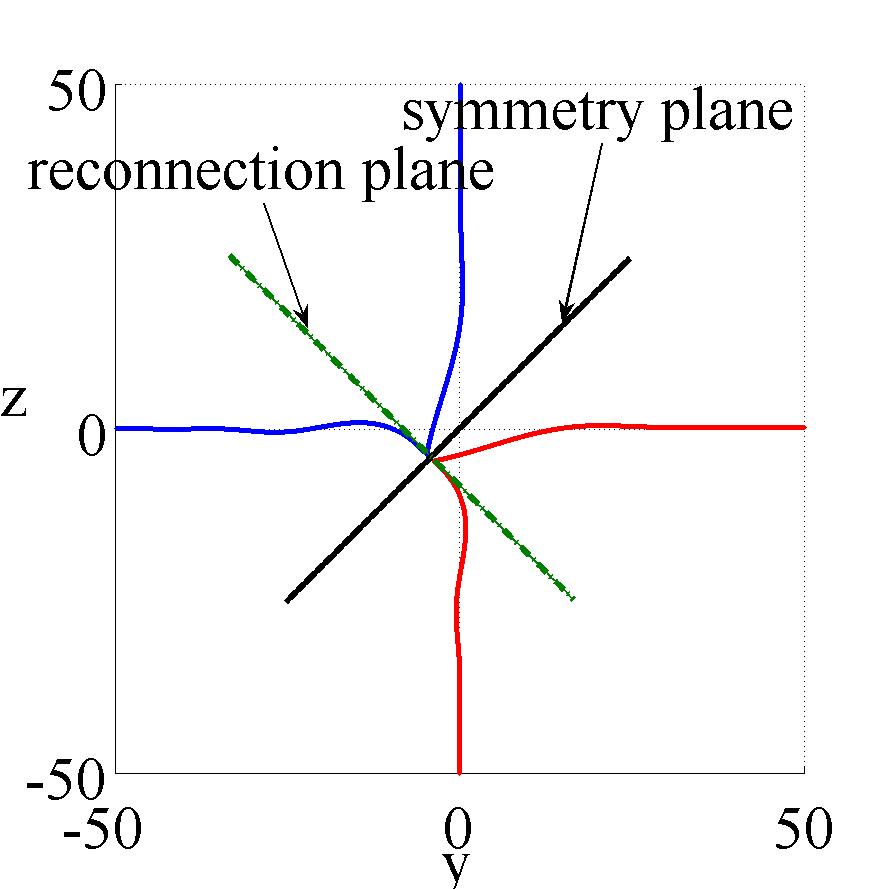}
\caption{\label{fig:yzplane} 
Sketch of the orthogonal reconnection around the reconnection point ${\bf x}_{r}$ 
at the reconnection time $t_r$ of the vortices projected onto the $y-z$ plane. 
The pre-reconnection trajectories approximately follow 
the $y-y_{mid}=0$ and $z-y_{mid}=0$ axes and the post-reconnection trajectories are the
red and blue lines. Two projected planes are indicated in black.  The 
{\it reconnection plane} $(z-z_r)=-(y-y_r)$ contains the tangents 
to the vortex lines $\bT_{y,z}$ and their curvature vectors $\bN_{y,z}$ 
both before and after reconnection, as well at the separation vector between 
the closest points $\bx_{y,z}$.
The {\it propagation/symmetry plane} $(z-z_r)=(y-y_r)$ which contains the 
bi-normals $\bB_{y,z}$ and the {\it direction of propagation} of the mid-point
$\bx_r(t)$ between the closest points of the vortices: $\bx_y$ and $\bx_z$.
From this perspective the angles along the {\it reconnection plane} before 
reconnection are shallow, and those after reconnection are sharp.}
\emini~~~\bminir{0.35}
\includegraphics[scale=.40]{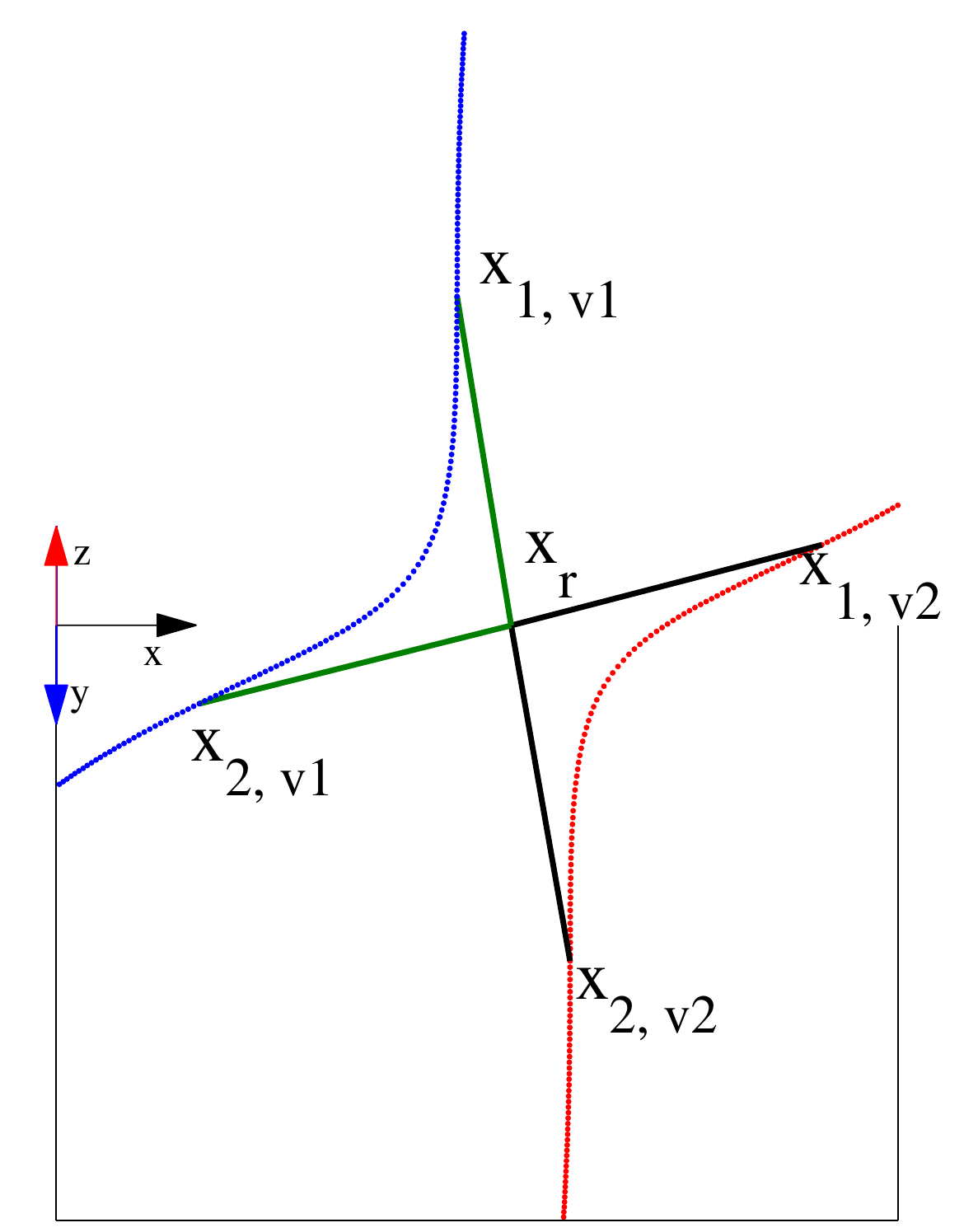}
\caption{\label{fig:NPplane}
Sketch using the Nazarenko perspective, that is looking down on the
{\it reconnection plane} along the {\it propagation/symmetry plane}, just before reconnection.
Four line segments (in three-dimensions) have been added 
that connect $\bx_r$, the mid-point between the vortices at closest approach, to
four points, each located $\Delta s$ units along the vortices from their respective
closest points.  Angles generated between adjacent segments as a function of 
$\Delta s$ are plotted in figure \ref{fig:angles2D}.  
}
\emini
\EFG

{\bf Sketches: 2D and 3D.} 
Figures \ref{fig:yzplane} and \ref{fig:NPplane}
provide two planar sketches at or near the reconnection time, with
the best reference point for each being
the mid-point between the closest points on the two vortices:
\EQL{eq:midpoint} \bx_r(\delta_0,t)=0.5(\bx_y(t)+\bx_z(t))\,. \EN
The sketch in figure \ref{fig:yzplane} projects 
the vortices at the reconnection time $t_r$ onto the $y-z$ plane around 
the point of reconnection: $\bx_r(\delta_0,t_r)$, along with projections of
two planes of interest, the {\it reconnection or osculating plane} and the
{\it propagation/symmetry plane}, which are defined in terms of the average 
Frenet-Serret basis vectors \eqref{eq:TNBav} in subsection \ref{sec:orthoorient}.
The second sketch in figure \ref{fig:NPplane} looks down the $x=0$, $y=z$
45$^\circ$ {\it direction of propagation} of $\bx_r(\delta_0,t)$ onto the
$\bT_{av}\times\bN_{av}$ \eqref{eq:TNBav} {\it reconnection plane}. 
Important features include:

\ITM\item Because figure \ref{fig:yzplane} is at $t=t_r$, the projections
of the planes and vortices all cross at $\bx=\bx_r(\delta_0,t_r)$, which
means that the red and blue curves trace both the pre- and post-reconnection
trajectories of the vortices, as follows:
\ITM
\item The trajectories before reconnection follow the curves 
parallel to the $y=0$ and $z=0$ axes that are half blue and half red. 
These are projections of the red and blue lines in figure \ref{fig:NPplane}.
\item The trajectories immediately after reconnection are indicated by the 
red and blue curves coming out of the {\it reconnection plane}.
\ITN
\item The two orthogonal lines through $\bx_r$ represent two planes:
\ITM\item The {\it reconnection plane}, defined by $\bT_{av}=0.5(\bT_y-\bT_z)$ and
$\bN_{av}=0.5(\bN_y-\bN_z)$ \eqref{eq:TNBav}. Before or after reconnection,
the separation vector $\bD=(\bx_z-\bx_y)/|\bx_z-\bx_y|\neq0$ is also in this plane. 
\ITM\item It is shown below that $\bT_{y,z}$ and $\bN_{y,z}$ swap at 
reconnection, so all of these basis vectors stay in this plane
after reconnection.\ITN
\item And the {\it propagation plane}, which contains the 
velocity of $\bx_r(\delta_0,t)$ and 
the average bi-normal $\bB_{av}=0.5(\bB_y+\bB_z)$ \eqref{eq:TNBav}. \ITN
\item The $x=0$, $y=z$ projection in figure \ref{fig:NPplane} 
is denoted the Nazarenko perspective or NP because it 
was used by linear model of \cite{NazarWest03}. 
\ITM\item That model tells us that $\bx_r=0.5(\bx_y+\bx_z)$ translates 
in the $(-y,-z)$ direction, motion that implies that the reconnection does not 
occur at the centre of the computational box. 
\item Note that for the points on either side of $\bx_r$, the tangents $\bT_{y,z}$ and 
curvature vectors $\bN_{y,z}$ are anti-parallel. The components that are not 
anti-parallel are directed out of the Nazarenko perspective.  
This also holds for the lines
across the central, green $\rho=0.05$ isosurface in figure \ref{fig:orthoT6-12}b. 
figure \ref{fig:TNBD}a shows how $\bT_y\cdot\bT_z$, $\bN_y\cdot\bN_z$ and 
$\bB_y\cdot\bB_z$ converge to this state as $t\rightarrow t_r$.
\item The lines drawn across the centre of  figure \ref{fig:NPplane} are used 
to determine the long-range angles discussed in section \ref{sec:orthoangles}.
\ITN
\item {\bf 3D by using fingers.} Cross your index fingers
while pointing their knuckles towards one another so they do not touch. 
\ITM\item Rotating this configuration reproduces figures 
\ref{fig:orthoazel}a, \ref{fig:orthoNP}a and \ref{fig:orthoT6-12}a.
\item Now move your fingers up, bending them as you do and bringing the
knuckles together.
\item This is how, the alignments of the Frenet-Serret frames at the points 
of closest approach in figure \ref{fig:TNBD} form. 
\item Now hold that configuration and rotate it to get the configurations of the 
sketches in figures \ref{fig:yzplane} and \ref{fig:NPplane} and the three-dimensional 
images in figures \ref{fig:orthoazel}b, \ref{fig:orthoNP}b and \ref{fig:orthoT6-12}b.
\ITN
\ITN

\subsection{Curvature \label{sec:orthocurve}}

Curvature, has played a central role in our understanding of quantum turbulence due to its 
use in predicting velocities in the law of Biot-Savart and the local induction 
approximation.  The connection between
these approximations for the velocities and the true dynamics of 
quantum fluids, as modeled by the Gross-Pitaevskii equations, would be in
how the gradient of the phase $\phi$ of the wavefunction $\psi$
is modified by the curvature of the vortex lines. 

In that context, could curvature profiles provide clues for 
the origins of the anomalous scaling exponents of the orthogonal separations?  
For example, if the local induction approximation is relevant, then a sudden increase in the
maximum curvature of the lines could explain the change in scaling.

To assess whether this is a possible explanation,
figure \ref{fig:orthocurv} plots profiles of the curvature along the vortex lines at several times,
and the curvature maxima $\kappa_{\rm max}$ as a function of time for the $\delta_0=3$ case. 
The other cases have similar
behaviour, including one slightly under-resolved $\delta_0=4$ case in a large $(32\pi)^3$ domain.

Two primary features should be noted. First, in figure \ref{fig:orthocurv}a there is 
very little asymmetry in $s$ about the points of cloest approach, both before and after reconnection. 
Second, in figure \ref{fig:orthocurv}b, there is some growth in $\kappa_{\rm max}$ for $t<t_r$, followed
by a sharp jump in $\kappa_{\rm max}$ at $t=t_r$, which then relaxes rapidly 
to the pre-reconnection values of $\kappa_{\rm max}$.

So there is some qualitative support for a local induction explanation change between the 
$\delta_{in}$ and $\delta_{out}$ scaling, but there is also inconsistency with the following:
\ITM\item Even pre-reconnection, the $\delta_{in}\sim (t_r-t)^{-1/3}$ scaling would require 
a stronger growth in $\kappa$ than is observed.
\item Post-reconnection, and after the curvature spike has relaxed, stronger curvatures than are
observed would be needed to maintain the $\delta_{out}\sim (t_r-t)^{-2/3}$.
\ITN

Therefore, one must conclude that a bigger picture is needed. Our proposal is to look at the 
alignments of their respective Frenet-Serret frames as another reason for the changes in scaling.
Looking first at the alignments of at the points of closest approach, 
then at distances away from those points \citep{Rorai12}.

Another possible explanation, if the local induction approximation is relevant, 
would be if the curvature maxima are stronger post-reconnection.
The inset in figure \ref{fig:orthocurv} does show that the maxima are slightly 
stronger post-reconnection, but this does not appear to be the
whole story.  To get a bigger picture we need to look beyond the curvature
profiles and see how the vortices are aligned with each other for points along
their entire trajectories. Starting with the alignments of their respective
Frenet-Serret frames at the points of closest approach, then for distances
away from those points.

\subsection{Orthogonal: Frenet-Serret orientation. \label{sec:orthoorient}}

Besides allowing us to calculate the curvature of the vortex lines, knowing
their trajectories allows us to calculate the Frenet-Serret frame \eqref{eq:FrenetS}.
This has been done for the $256^3$, $\delta_0=3$ calculation for
all times and provides quantitative support for the describing the local frame
at the reconnection point $\bx_r$ in terms of the {\it reconnection plane} 
and the {\it propagation plane} indicated in figure \ref{fig:yzplane}. 

Figures \ref{fig:TNBD} and \ref{fig:D-TNB} show the following evolution
of the Frenet-Serret frames at the closest points $\bx_{y,z}$: 
\ITM\item For $t<t_r$, pre-reconnection: 
\ITM\item figure \ref{fig:TNBD}a shows that the vorticity direction or tangent vectors 
$\bT_{y,z}$, the curvature vectors $\bN_{y,z}$  and the bi-normals $\bB_{y,z}$ converge
to their opposites as $t\rightarrow t_r$, both before and after reconnection.
\item figure \ref{fig:TNBD}b shows that $\bT_{y,z}$, $\bN_{y,z}$ and their unit 
separation vector $\bD=(\bx_z-\bx_y)/|\bx_z-\bx_y|$ all lie in the same plane, 
the {\it reconnection plane} in figure \ref{fig:yzplane}.
\item The bi-normals $\bB_{y,z}$ define the {\it propagation plane}.
\item The useful averages of the Frenet-Serret frames between $\bx_y$ and $\bx_z$
are these: 
\EQL{eq:TNBav}\hspace{-6mm}
\begin{array}{clclc}\medskip
\circ & \bN_{av} & = & 0.5(\bN_y-\bN_z) & \quad\mbox{and is parallel to}~\bD(t)\\ \medskip
\circ & \bT_{av} & = & 0.5(\bT_y-\bT_z) & \quad\mbox{and is perpendicular to both}~\bN_{av}~{\rm and}~\bD.\\ \medskip
\circ & \bB_{av} & = & 0.5(\bB_y+\bB_z) & \quad\mbox{and is perpendicular to}~\bT_{av},~
\bN_{av}~{\rm and}~\bD. \\
\end{array} \EN
\item That is: Subtract the tangent and curvature vectors because they
become anti-parallel as $t\nearrow t_r$, and add the bi-normals because they
are parallel as $t\nearrow t_r$. 
\item These alignments between $\bT_{av}$, $\bN_{av}$ and $\bB_{av}$ with 
$\bD$ form in the early stages, long before the reconnection at $t_r=8.9$. 
\ITN
\item The post-reconnection Frenet-Serret flip: 
\ITM\item Figure \ref{fig:D-TNB} shows that the directions of $\bT_{av}$ and $\bN_{av}$ swap and $\bD$ rotates by 
90$^\circ$ so that all three are still in the {\it reconnection plane} 
with the same relations to one another. 
\item $\bB_{av}$ remains orthogonal to the {\it reconnection plane}. 
\ITN\ITN

\BFG
%\bminic{0.62}
%\includegraphics[scale=.26]{TNBDtime17feb16.jpg}
\includegraphics[scale=.26]{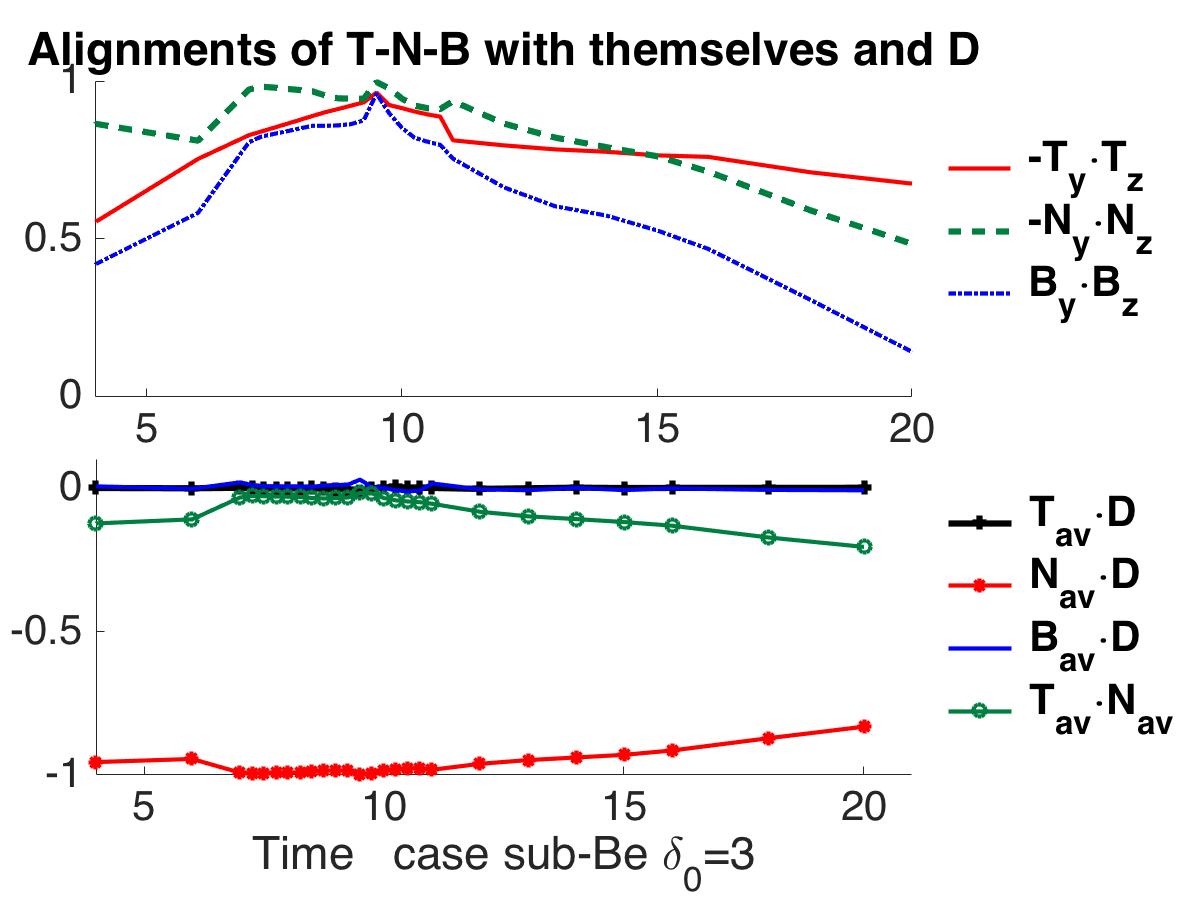}
\begin{picture}(0,0)
\put(-175,150){\Large a)}
\put(-175,65){\Large b)}
\end{picture}
%\emini~~\bminic{0.36}
\caption{\label{fig:TNBD} These two frames show how the basis vectors of the Frenet-Serret 
frames \eqref{eq:FrenetS} at the closest points $\bx_{y,z}$ of the $\delta_0=3$, $256^3$ 
calculation are aligned and help define 
the {\it reconnection plane} and the {\it propagation plane} used in figures
\ref{fig:yzplane} and \ref{fig:NPplane}.
a): The alignments between the Frenet-Serret components at $\bx_{y,z}$. 
For $t<t_r$ the tangent vectors $\bT_{y,z}$ and curvature vectors 
$\bN_{y,z}$ become increasingly anti-parallel as the reconnection time is approached 
while the bi-normal vectors $\bB_{y,z}$ become increasingly aligned.  
These trends are are reversed for $t>t_r$.
b): The alignments between the averages over $\bx_{y,z}$ of the Frenet-Serret frames 
defined by \eqref{eq:TNBav} with the separation 
vector $\bD$ between $\bx_y$ and $\bx_z$.  }
%\emini
\EFG

\BFG\bminic{0.57}
\includegraphics[scale=.20]{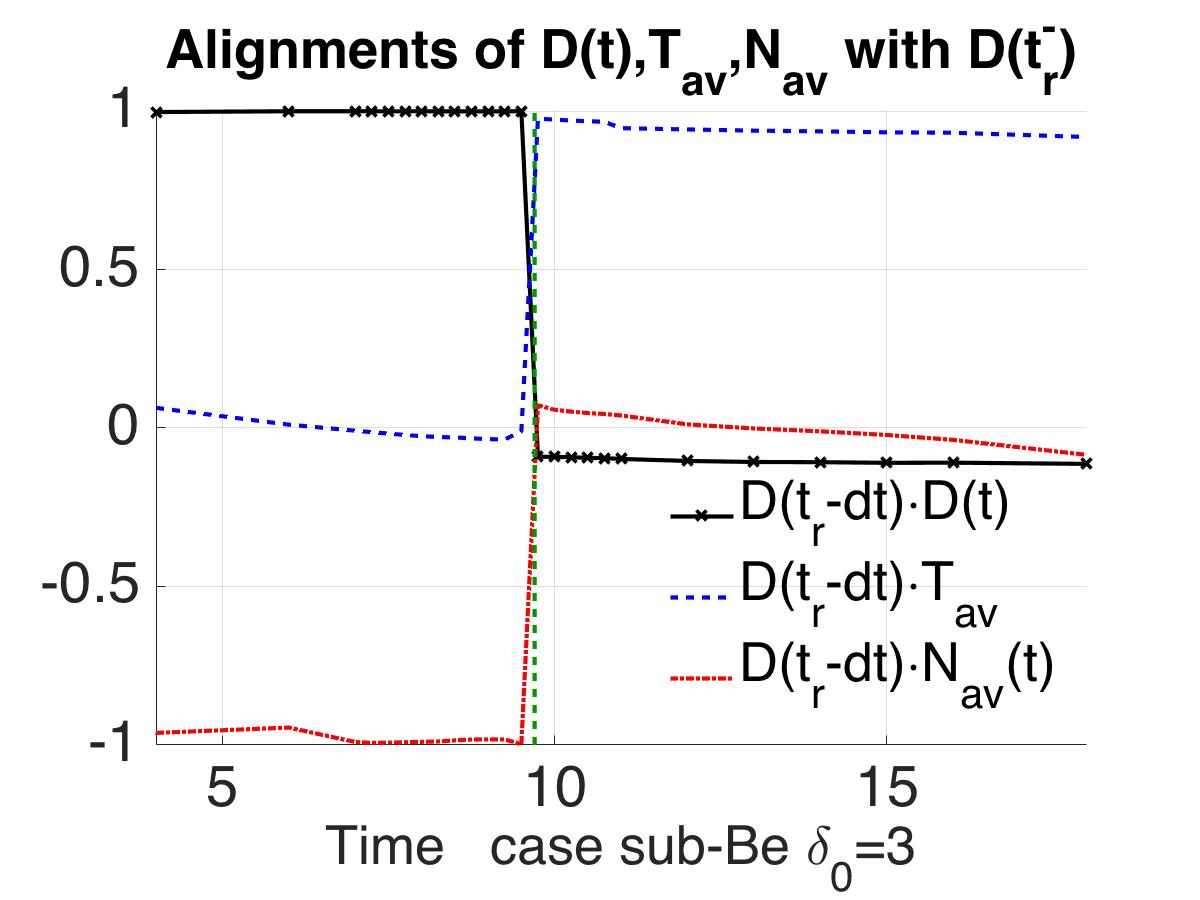}
\emini~~\bminic{0.40}
\caption{\label{fig:D-TNB} 
For the $\delta_0=3$, $256^3$ calculation, the alignments over time of 
$\bD(t_r^-)=\bD(t_r-dt)$ with $\bT_{av}$ and $\bN_{av}$, the averages \eqref{eq:TNBav} 
over the tangent and curvature components of the Frenet-Serret frames at
the closest points $\bx_{y,z}$. $\bB_{av}\cdot\bD(t_r^-)\approx0$ for all times.
$\bD(t_r^-)$ is the separation direction at $t=8.9$, just before reconnection.
}
\emini
\EFG

The alignments of the components of the Frenet-Serret frames at $\bx_{y,z}$ and
the separation of these points $\bD$ is significantly different than their
alignments for the anti-parallel case in Sec. \ref{sec:antipresults}. Comparisons
are discussed in the Summary in Sec. \ref{summary}.

%%%%%%%%%%%%%%%%%%
\subsection{Angles at reconnection \label{sec:orthoangles}}

\BFG
\includegraphics[scale=.40]{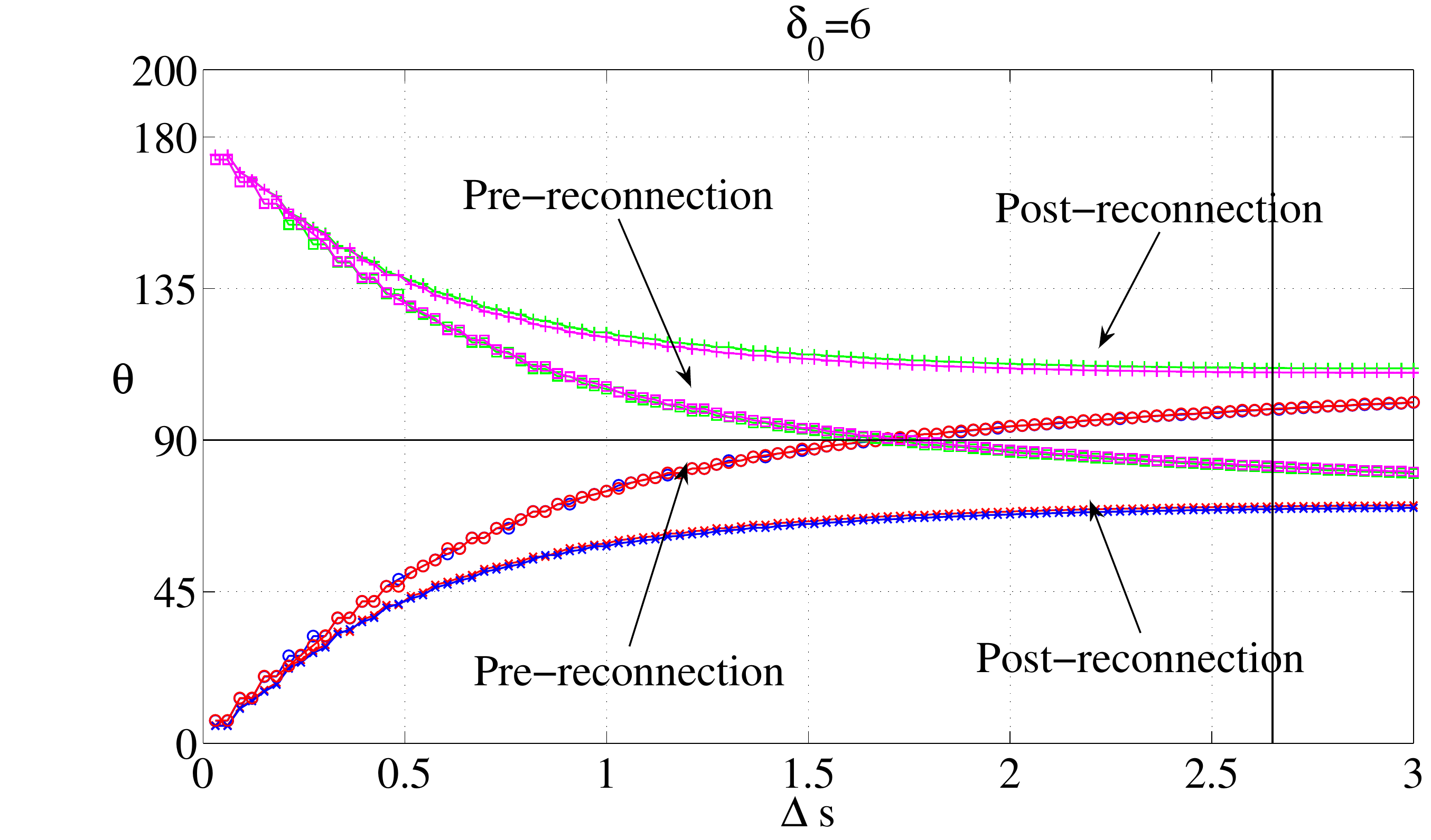}\\
\caption{\label{fig:angles2D} 
Angles between the segments in 3D space displayed in figure \ref{fig:NPplane} 
as a function of $\Delta s$ for $\delta_0=6$ and $t=63$ (pre-reconnection) and $t=65$ (post-reconnection). The angles between the arms of the same vortex, triplets [$\bx_{1,v1},~\bx_r,~\bx_{2,v1}$] and [$\bx_{1v2},~\bx_r,~\bx_{2v2}$], are indicated by circles (blue and red) before reconnection, and stars (blue and red) after reconnection. The angles sum to 360$^\circ$ as $\Delta s\rightarrow0$, implying that the segments 
near the reconnection point $\bx_r$ lie on a plane, specifically, the reconnection plane in figure \ref{fig:NPplane}
The sum of the angles grows with $\Delta s$ and is about $\sim362^\circ$ for $\Delta s\sim2.6$ (vertical black line), which is consistent with the arms of reconnecting vortices becoming progressively more convex and hyperbolic as they move away from $\bx_r$, both before and after reconnection.
The rapidly changing behavior for small $\Delta s$ is a kinematic result of how
the lines and angles indicated in figure \ref{fig:NPplane}
were chosen and is not significant. What could be more significant are 
the differences in the angles for intermediate $\Delta s$ and their influence upon
any Biot-Savart contributions to the velocities. 
}
\EFG

What additional dynamics might help us identify the differences between the
orthogonal and anti-parallel separation scaling laws?

%While the differences in the local Frenet-Serret frames between the
%orthogonal cases and the anti-parallel case seem to be behind the differences
%in their respective separation scaling laws, can this be understood
%dynamically?

One place to look is larger-scale alignments and long-range interactions. 
While the underlying Gross-Pitaevskii equations are local, the existence of 
the vortex structures means that such alignments should exist and should influence
the local motion of the $\rho\equiv0$ lines to the degree that the law of
Biot-Savart can be applied.  With the goal of identifying any such long-range
interactions, this section determines the evolution of pre- and post-reconnection 
angles between points on the extended structures.  

Relevant angles can be defined between the arms of the reconnecting vortices 
around the reconnection point as follows. 
\ITM\item[({\it i})] From the points of closest approach $\bx_{y,z}(t)$, 
define $\bx_r(t)=0.5(\bx_y(t)+\bx_z(t))$.
\item[({\it ii})] Move $\pm\Delta s$ along the arms of the vortices from 
$\bx_{y,z}(t)$ and identify four new points:  
${\bf x}_{1,y}$, ${\bf x}_{2,y}$, ${\bf x}_{1,z}$ and ${\bf x}_{2,z}$,
illustrated in the Nazarenko perspective sketch in figure %\ref{fig:NPplane}.
\ref{fig:orthoazel} 
\item[({\it iii})] To get a fully three-dimensional perspective, note that these points lie
on outstretched arms such as in figure \ref{fig:orthoazel}b.
\item[({\it iv})] Connect the four points with ${\bf x}_{r}$ to form an extended
three-dimensional frame then calculate the angles $\theta_i$ between these four vectors. 
\item[({\it v})] Plots of $\theta_i(\Delta s)$ show {\it qualitatively} similar variations
independent of $\delta_0$ with these properties:
%Results for $\delta_0=2,3,4,5,6$ are reported in Fig. \ref{4angles}, \ref{4angles2}.
\ITM\item The sum of the $\theta_i$ grows as $\Delta s$ increases, starting from
$\sum_i \theta_i(\Delta s=0) =360^{\circ}$. %(by construction).
\item This shows that the inner ($\Delta s\approx0$) structure is a plane. 
\item The vertical line in figure \ref{fig:angles2D} represents the $\Delta s_p$ 
for which $\sum_i \theta_i(\Delta s) =362^{\circ}>360^\circ$, indicating that
the structure is mildly hyperbolic.
\item $\sum_i \theta_i(\Delta s)$ increases with $\Delta s$, implying that the global 
structure is convex or hyperbolic, the opposite of a structure with an acute angles 
such as a pyramid.
\item The primary quantitative difference as $\delta_0$ decreases is that the $\theta=90^\circ$ 
pre-reconnection cross-over, at $\delta s=1.5$ for $\delta_0=6$, decreases. 
\item In addition, the geometry becomes more hyperbolic. That is $\sum_i \theta_i(\Delta s)$ increases 
as $\delta_0$ decreases with $\sum_i \theta_i(\Delta s)=370^\circ$ for$\delta_0=3$. 
\item The inner angles, that is the angles between the triplets [$\bx_{1,v1},~\bx_r,~\bx_{2,v1}$] 
and [$\bx_{1v2},~\bx_r,~\bx_{2v2}$], are larger than 90$^\circ$ before reconnection and smaller after. 
The difference is about 20$^\circ$ in all of the calculations and would be consistent 
with the observed slower approach and faster separation.
\item When looking down at the {\it reconnection plane} in either 
figures \ref{fig:NPplane} or figure \ref{fig:orthoNP}b, 
do not forget that the centre of the {\it reconnection plane}
is moving along the $y=z$ direction of the {\it propagation plane} and simultaneously 
dragging or pushing the extended arms as it moves, as in figure \ref{fig:orthoazel}b.
\ITN
\ITN

Understanding these features could provide us with some hints about the origins of
the anomalous scaling laws.  One hint could be the different angles at intermediate 
scales, $0.5\leq\Delta s\leq 2.5$.  Pre-reconnection, the this span has approximately 
retained its original $\theta\sim 90^\circ$ orthogonal alignment. Post-reconnection the 
angles over this span still sums to approximately $362^{\circ}$ degrees, but the angles on either side of the reconnection have jumped by $\pm20^{\circ}$. The sudden jump is due to
how the directions of the tangent and curvature vectors swap at reconnection.
Using the sketch in figure \ref{fig:NPplane}, this means that prior to reconnection
the angle in 3D space is between the triplets 
[$\bx_{1,v1},~\bx_r,~\bx_{2,v1}$] and [$\bx_{1v2},~\bx_r,~\bx_{2v2}$] 
and afterwards between the triplets
[$\bx_{1v1},~\bx_r,~\bx_{1v2}$] and [$\bx_{2v1},~\bx_r,~\bx_{2v2}$].
Similar behaviour is seen for all the $\delta_0$ cases. It is visualised for 
$\delta_0=3$, using two perspectives in figures \ref{fig:orthoazel} and \ref{fig:orthoNP}.
\AD{
Similar behaviour is seen for all the $\delta_0$ cases. one way to visualise 
the effect of this swap is to consider its effect upon the new angles for the
$\delta_0=3$ frames in figures \ref{fig:orthoazel} and \ref{fig:orthoNP}.}

To goal is find a model that links the sudden swaps in local alignments in 
figures \ref{fig:TNBD} and \ref{fig:D-TNB} with the nonlocal changes in 
figure \ref{fig:angles2D} and from that explains the anomalous orthogonal reconnection  
scaling.  This model must also accommodate the scaling of anti-parallel reconnection, 
which obeys the expected dimensional scaling both before and after reconnection. 
The final discussion in section \ref{sec:conclude} addresses what might be required.

%%% RMK END
%\input{Anti23jul14}
\section{Anti-parallel results: approach, separation, curvature }\label{sec:antipresults}

In contrast to the orthogonal vortices, whose scaling laws
do not obey expectations, it will now be shown that the scaling of initially 
anti-parallel vortices obeys those expectations almost completely.

Figure \ref{fig:antipstreamrhoT0T4} illustrated the overall structure of our 
anti-parallel case before reconnection at $t=0.125$ and after reconnection
at $t=4$, where $t_r\approx2.44$. The very low density $\rho=0.05$ isosurface 
is the primary diagnostic for the vortex lines and as in \cite{Kerr11}, 
the post-reconnection density isosurfaces are developing a second
set of reconnections near $y=\pm5$ from which the first set of vortex rings will 
form. Eventually two stacks of vortex rings on either side of $y=0$ should form
from the additional waves along the original vortex lines. 
Large values of the gradient kinetic energy \eqref{eq:GPenergies} and the magnitude 
of the pseudovorticity \eqref{eq:pseudov} are shown using two additional 
isosurfaces and the trajectories of the $\rho\approx0$ pseudo-vorticity 
are traced with thick lines. 

Figure \ref{fig:antipstreamrhoTtr} provides two perspectives of the structure 
approximately at the reconnection time using the same isosurfaces and vortex lines 
as in figure  \ref{fig:antipstreamrhoT0T4}, except there are now two sets of vortex 
lines. figure \ref{fig:antipstreamrhoT0T4}a shows two red curves seeded using
$\rho\approx0$ points on the $y=0$, $x-z$ {\it perturbation plane} and two
blue curves seeded using
$\rho\approx0$ points on the $z=0$, $x-y$ {\it dividing  plane}, representing
respectively the pre- and post-reconnection trajectories. 

As for the $\delta_0=3$ orthogonal case in figure \ref{fig:orthoT6-12}b, 
the reconnection time isosurfaces have an extended zone of very low density 
around the reconnection point and a strong isosurface of the gradient kinetic 
energy outside this zone.
Large values of $|\omega_\rho|$ are also outside the reconnection zone.

\BFG
\hspace{-0mm}
\includegraphics[scale=.18]{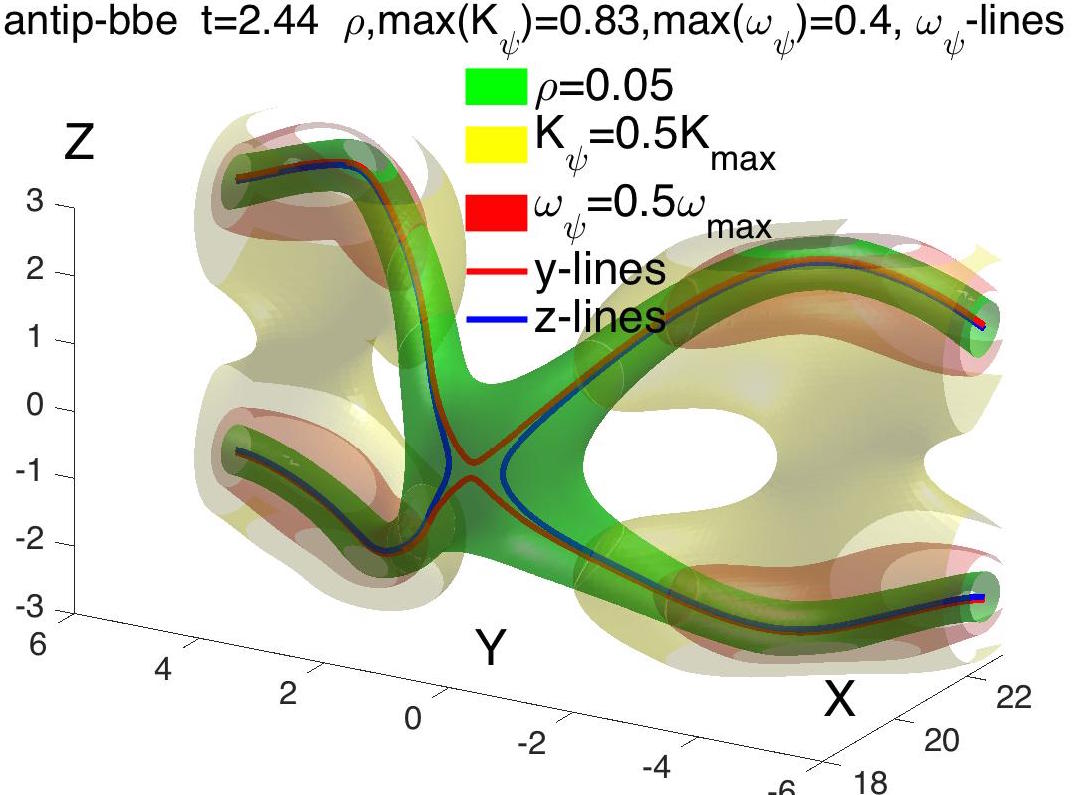}
\begin{picture}(0,0)
\put(-180,123){\LARGE(a)}
\end{picture}
\includegraphics[scale=.18]{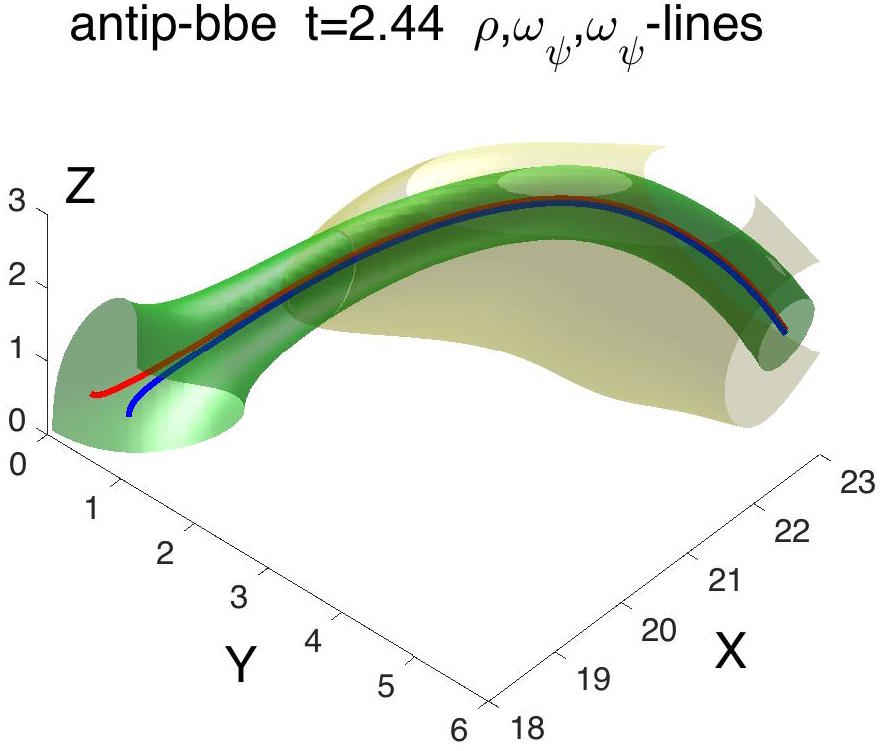}
\begin{picture}(0,0)
\put(-100,110){\LARGE(b)}
\end{picture}
\caption{\label{fig:antipstreamrhoTtr} Density, kinetic energy 
$K_{\psi}=0.5|\nabla\psi|^2$ \eqref{eq:GPenergies} and pseudo-vorticity 
$|\nabla\psi_r\times\nabla\psi_i|$ isosurfaces plus vortex lines a bit after
the time of the first reconnection,  $t=2.4375\gtrsim t_r=2.344$ 
from two perspectives with $\max(K_{\psi})=0.83$. $\max|\omega_\rho|=0.4$ 
for all the times. (a) shows four vortex 
lines, with each pair starting at the points of nearly zero density on the 
$y=0$ and $z=0$ symmetry planes respectively.  (b) gives a second perspective 
that looks into
a quadrant through the two symmetry planes and is designed to demonstrate that locally,
around $y=z=0$, the shape of the $\rho=0.05$ isosurface is symmetric on the two
symmetry planes.  
}
\EFG

Figure \ref{fig:antipinout} presents several possible fits for 
the pre- and post-reconnection separation scaling laws, with 
both the pre-reconnection incoming and post-reconnection outgoing separations 
following the predicted dimensional scaling of $\delta\sim |t_r-t|^{1/2}$ 
\eqref{eq:deltat}, unlike the orthogonal cases just shown. Furthermore,
for $\Delta t=|t-t_r|<0.5$, the $\delta_{\rm in}(t_r-t)$ and 
$\delta_{\rm out}(t_r-t)$ are almost mirror images of each other. Although
that is not the case for larger $\Delta t$.  

This suggests that unlike the orthogonal case, where we have attempted to relate the
asymmetric scaling laws to asymmetries in the underlying structure, for the 
anti-parallel case we want to identify physical symmetries that would predispose 
the scaling laws to be temporally symmetric and follow the dimensional prediction.

The purpose of giving two perspectives near the reconnection time in 
figure \ref{fig:antipstreamrhoTtr} is to clarify the physical symmetries at this time.
The overall structure in figure \ref{fig:antipstreamrhoTtr}a shows how all four legs 
of vorticity converge on the $y=z=0$ line where the $x-z$ or $x-y$ symmetry planes 
cross.  Then figure \ref{fig:antipstreamrhoTtr}b looks at the structure from the 
interior of the zone of nearly zero density from the perspective of the $\rho\approx0$, 
$y-z$ reconnection plane upon which the vortices reconnect. In a manner analogous to 
when the orthogonal vortices reconnect in figure \ref{fig:orthoT6-12}, the tangent
and normal vectors swap directions, with pre-reconnection tangent in $y$ becoming the 
post-reconnection normal, and the normal in $z$ becoming the tangent. This results in 
the butterfly trajectories around the $y=z=0$ line in figure \ref{fig:antipstreamrhoTtr}a. 

A useful, but not perfect, way to characterise the resulting structure is to use 
the proposal by \cite{deWaeleAarts94}, based on their Biot-Savart calculation of 
anti-parallel quantum vortices, that near reconnection the vortices form an 
equilateral pyramid. 
The pyramid for figure \ref{fig:antipstreamrhoTtr}a can be formed by straightening 
the 2 red and 2 blue curved legs of vorticity surrounding the $y=z=0$ line in $x$ 
to $(x,y,z)\approx(18,0,0)$, slightly to the left of the $\rho=0.05$ isosurface. 
These extensions would start from tangents to the red and blue lines. 
The angles $\theta_{yz}$ of these tangents with respect to the symmetry planes 
would define the sharpness of the tip of pyramid and depend on where the tangents
are taken from. $\theta_{y,z}=45^\circ$ is obtained for  $|y|=|z|\leq2$, 
which is where the vortices begin to bend back upon themselves. 

Why don't the pseudo-vortex lines continue to bend, or kink, until a sharp tip with
$\theta_{y,z}=45^\circ$ is obtained? Figure \ref{fig:antipcurv} provides the
clues by directly plotting the curvatures and the inclinations of the tips of 
the $y$ and $z$ pseudo-vortex lines, which are
the points of closest approach in figure \ref{fig:antipcurv}a.
The curvatures are determined by \eqref{eq:curvature} and the inclinations come 
from $\bN=(n_x,n_y,n_z)$, the direction of the curvature from \eqref{eq:normal}.
Independent of whether one considers the red $y$-vortices or blue $z$-vortices,
their maximum curvatures have nearly the same upper bounds and the same maxima 
as $n_{y,z}/n_x$, with 
\EQL{eq:tipangles}\begin{array}{cl}
n_z/n_x\nearrow6 \quad{\rm as}~~ t\nearrow t_r & \quad {\rm for}~~ t>t_r\quad{\rm and} \\
n_y/n_x\searrow6 \quad{\rm as}~~ t\searrow t_r & \quad {\rm for}~~ t_r>t \,.\end{array}\EN
$\tan\theta_\pm=6$ corresponds to 
$80^\circ$, not the $45^\circ$ angles of a pyramid.  

This means that as reconnection is approached, the direction of the curvature
$\bN$ begins to be parallel to their separation $\bD$, very reminiscent of what 
has been found 
for the orthogonal vortices as $t\rightarrow t_r$ in figures \ref{fig:TNBD} and
\ref{fig:D-TNB}. Which also means that the directions of the curvature and
tangents nearly swap during reconnection, which is also similar to, but not
exactly the same as, the orthogonal cases. And not what a true pyramid with
a sharp tip would do.

What is probably more important for getting the dimensional scaling for $\delta(t)$
is that both sets of curves are temporally symmetric.
That is, $\kappa(t_r-t)\approx\kappa(t-t_r)$ and
$n_z(t_r-t)/n_x(t_r-t)\approx n_y(t-t_r)/n_x(t-t_r)$. 
Figure \ref{fig:antipcurv}b emphasises this further by showing the dependence of 
$\kappa_y$ and $\kappa_z$ on their arclengths $s$ at times just before and
after the estimated reconnection time of $t_r=2.34$.

Finally, let us go back to the three-dimensional structures. First, consider how
the $t<t_r$ vortex lines in figure \ref{fig:antipstreamrhoT0T4}a evolve into 
the pinched red $y$-vortex lines in figure \ref{fig:antipstreamrhoTtr}a  along a path 
in the $x-z$ symmetry plane with
$(n_x,0,n_z)(t)=(\cos\theta_-,0,\sin\theta_-)\nearrow (\cos80^\circ,0,\sin80^\circ)$.
Then note that the process is reversed for $t>t_r$ as the pinch in the 
blue $z$-vortex lines relaxes in $z$ as $\delta_+\searrow 0$ along a path 
in the $x-z$ symmetry plane from its most extreme orientation:
$(n_x,n_y,0)(t)=(\cos\theta_+,\sin\theta_+,0)\searrow (\cos80^\circ,\sin80^\circ,0)$.
Therefore the evolution of the in and out angles is temporally symmetric, as hoped for, 
even though the tip is not a true pyramid \eqref{eq:tipangles}.

These results differ from the anti-parallel case in \cite{Zuccheretal2012}, which
shows a slower approach and faster separation, similar to the scaling observed
for a recent anti-parallel Navier-Stokes calculation \citep{HussainDuraisamy11}.
There could be several reasons for the
differences. First, the initial perturbations to the anti-parallel trajectories in
\cite{Zuccheretal2012} are pointed towards one another, and 
not in the direction of propagation as here.
Another difference is that their periodic boundaries in $y$ (the vortical direction)
are relatively close, unlike that direction here. In \cite{Kerr11} and in a recent 
set of Navier-Stokes reconnection calculations \cite{Kerr13}, the advantages 
of making that direction very long have already been discussed. 

\BFG
\bminil{0.55}
\hspace{-10mm}
\includegraphics[scale=.22]{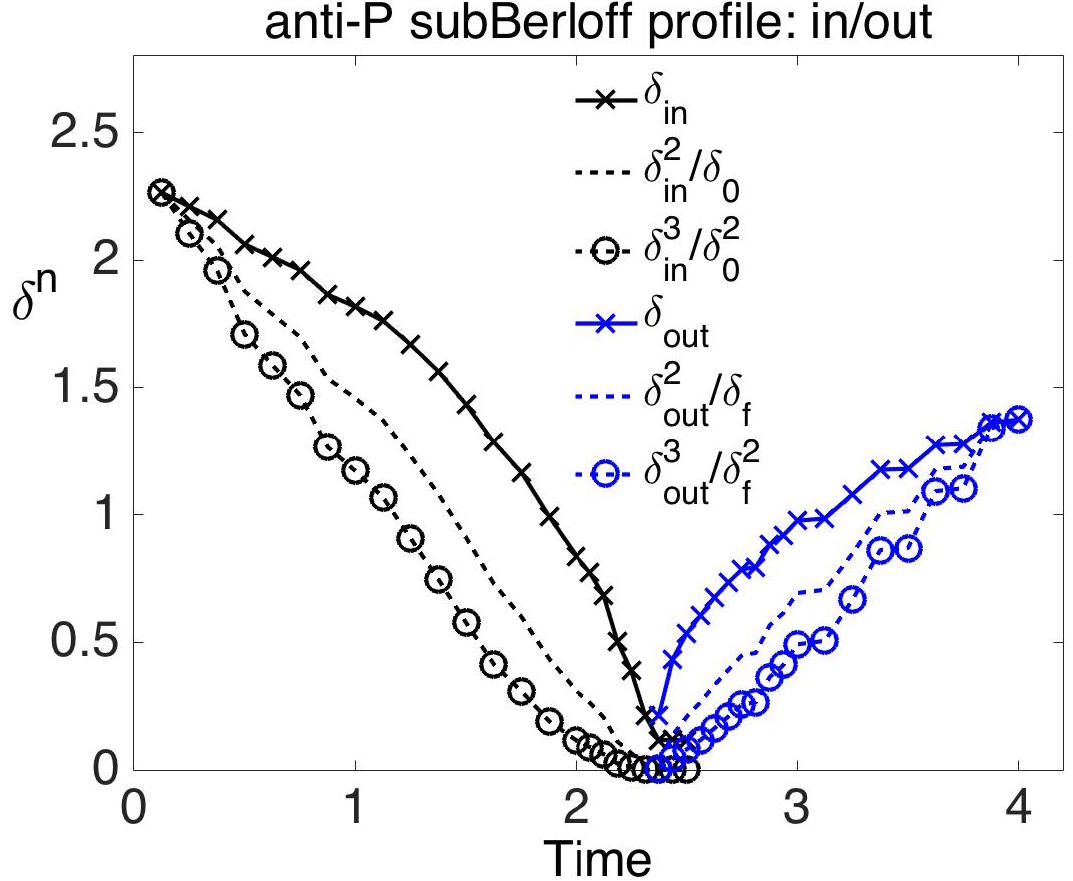} 
\emini~~~~\bminir{0.42}
\caption{\label{fig:antipinout}Separations near the time of the first reconnection 
of the anti-parallel vortices.  The positions
are found using the isosurface method on the symmetry planes. Both pre- 
and post-reconnection curves ({\it in} and {\it out})
follow $\delta_0\sim(t_r-t)^{1/2}$ most closely. }\emini
\EFG

\BFG
\includegraphics[scale=.18]{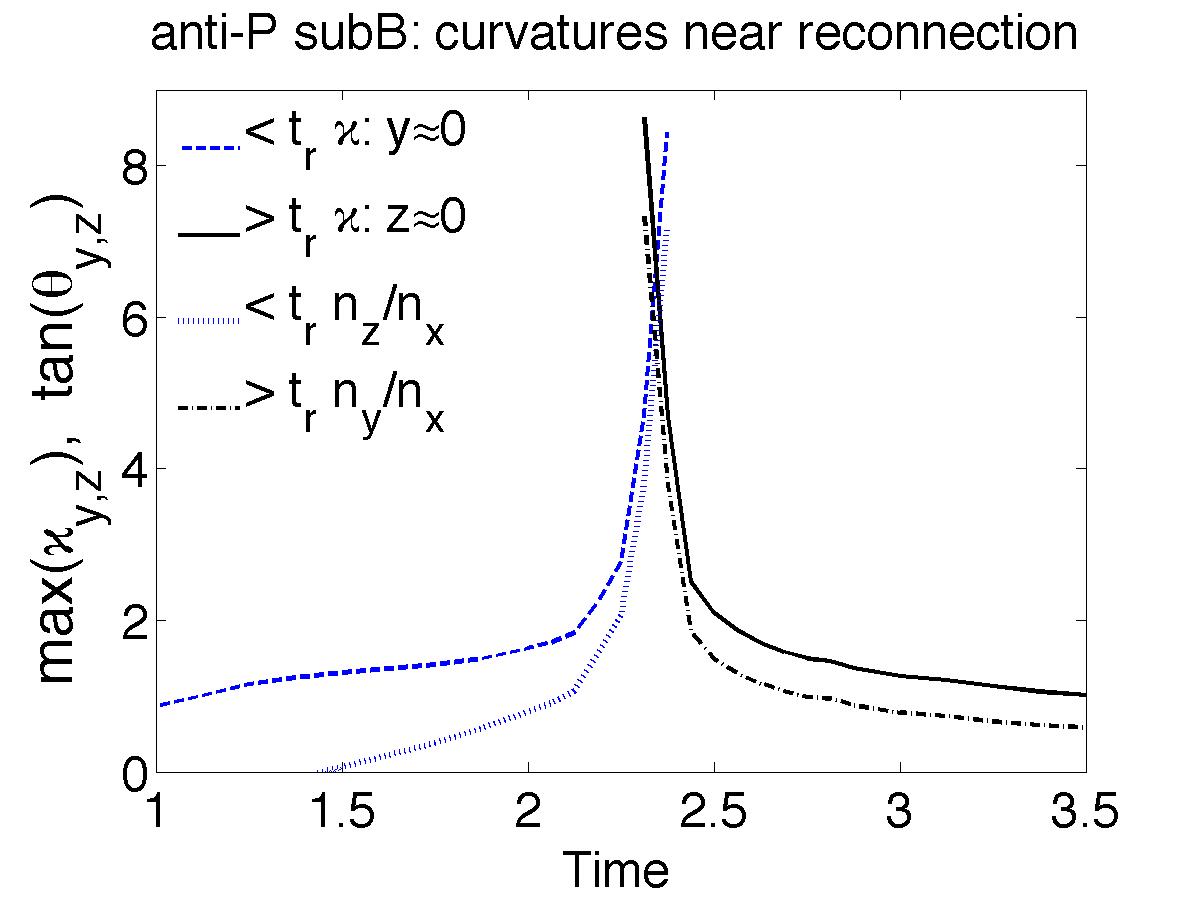}
\begin{picture}(0,0)
\put(-144,66){\LARGE(a)}
\end{picture}
\includegraphics[scale=.18]{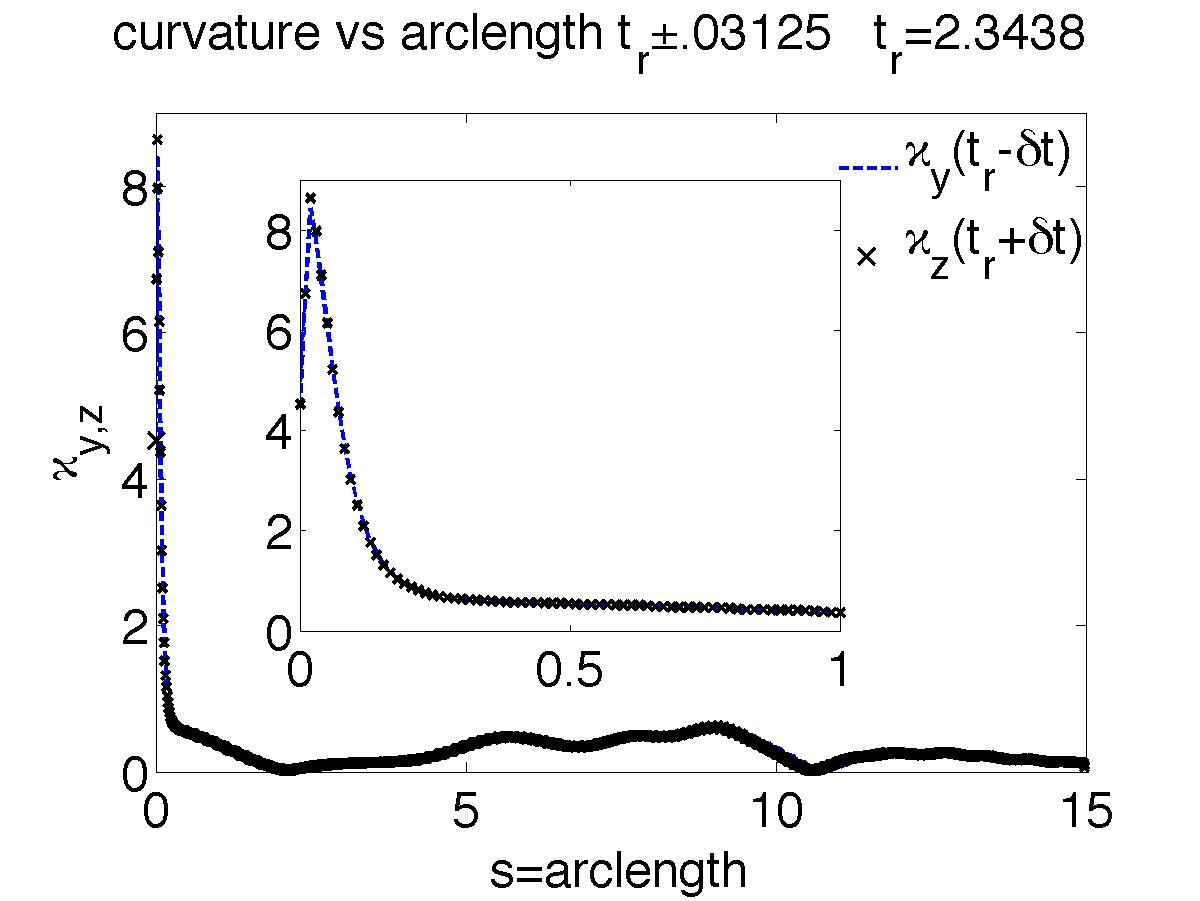}
\begin{picture}(0,0)
\put(300,90){\LARGE(b)}
\end{picture}
\caption{\label{fig:antipcurv}a:
%\begin{verb}
%###
%\end{verb}
Maximum curvatures, $\kappa_{y0}$, $\kappa_{z0}$, versus time.  
$\kappa_{y0}$ and $\kappa_{z0}$ are always on symmetry planes.  The
approach before reconnection and separation after reconnection are similar.
Reconnection time is where the curves cross, $t_r\approx2.344$.
b: Curvatures as a function of arclength along the vortices
near the reconnection time, with the profile from the new $z$ vortex taken
a bit before $t_r$ and the profile of the old $y$ vortex a bit after $t_r$.  
The curavatures near $y=z=0$
are similar, which drives similar approach and separation velocities.
}
\EFG
\newpage
\section{Summary \label{summary}}

The reconnection scalings of two configurations of paired vortices, orthogonal
and anti-parallel, have been found to have different scaling exponents.  

{\bf For the anti-parallel case,} the temporal scaling of both the pre- and 
post-reconnection separations obey the dimensional prediction, 
$\delta_\pm(t)\sim A\sqrt{\Gamma|t_r-t|}$ and the arms of the vortex pairs as the 
reconnection time is approached form an equilateral  pyramid with a smooth tip, 
which is in most respects qualitatively similar to the prediction of a Biot-Savart 
model \citep{deWaeleAarts94}.  Around the smooth tip the 
curvatures $\bN$ and separation $\bD$ are nearly parallel and as a result
the directions of curvature $\bN$ and the tangents $\bT$ almost swap during
reconnection.

{\bf The orthogonal cases}, in contrast, show 
asymmetric temporal scaling with respect to the reconnection time $t_r$.
For $t<t_r$, $\delta_-(t)\sim A_-(\Gamma|t_r-t|)^{1/3}$ and for $t>t_r$, 
$\delta_+(t)\sim A_+(\Gamma|t_r-t|)^{2/3}$, where the coefficients
$A_{\pm}$ are independent of the initial separation $\delta_0$. At $t\approx t_r$,
the reconnecting vortices are anti-parallel, with the vortices interacting in a
{\it reconnection plane} that contains the tangent and curvature vectors 
of both vortices as well as their separation vector. This results in the
directions of curvature and vorticity swapping during reconnection.

{\bf Two innovations} There are two innovations that allow these calculations to generate 
clean scaling laws.
One is an initial core profile that either minimises the formation of
secondary waves by the interacting vortices, or absorbs these waves. 
The second innovation is a way to trace the vortex lines that 
minimises the need to identify computational cells with small values of the density.

\subsection{Contrasting geometries \label{sec:summarygeometries}}

While this paper has emphasised the differences between the reconnection scaling
for the orthogonal and anti-parallel cases, a few similarites need to be noted when 
$t\approx t_r$ and $\bx\approx\bx_r$. That is within the reconnection zone in both time 
and space. First, within this zone the curvature vectors in both cases tend to align 
with the separation between the vortices and the opposing vortices are anti-parallel. The 
skew-symmetric alignment of the reconnecting orthogonal vortices seems to be sufficient
for imposing this local property.  Post-reconnection, in both
cases the curvature and tangent directions swap, or nearly so in the
anti-parallel case.

This also means that in neither case does a pyramid form in the zone immediately
around $\bx_r$.  Nor do any of the
fixed point solutions identified in \cite{Meichle12} form. 
%###

However, further from $\bx_r$, the situation is different. For the orthogonal cases,
angles between the vortices imply a convex or hyperbolic structure. In
the anti-parallel case, a pyramid forms with nearly acute angles.

Let us summarise the additional key features of the orthogonal cases.

\paragraph{\bf Orthogonal} From an early time, the closest points of the originally
orthogonal vortices become locally anti-parallel and their respective curvature 
vectors become anti-aligned with the line of separation.  Combined, this implies that the 
local bi-normals for each line are nearly parallel and do not point in the 
direction of separation.  At reconnection,
the directions of the vorticity and curvature swap, and the sign of the bi-normal
reverses.  All of this is in a {\it reconnection plane} defined by the 
averages of the curvature and vorticity directions at the points of closest 
approach. 

Another useful perspective is the Nazarenko perspective in figure \ref{fig:orthoNP}, 
which is along a
45$^\circ$ angle in the $y-z$ plane. From this perspective, the vortices are always
distinct, without any loops, and one can see that the pre-reconnection vortices
approach the reconnection from one direction, 
and post-reconnection vortices separate in another. This perspective is used for 
finding non-local alignments and angles as in figure \ref{fig:angles2D}, which
show that the global alignment of the initial orthogonal vortices
is hyperbolic.

While the three-dimensional graphics for our orthogonal cases are qualitatively similar
to the equivalent Gross-Pitaevskii density isosurfaces in \cite{Zuccheretal2012}, our 
interpretation of the underlying geometry is different.  \cite{Zuccheretal2012} conclude 
that the deviations of Gross-Pitaevskii separations from the dimensional prediction is 
only near the reconnection and probably due to the rarefaction waves they report. 
In contrast, our analysis shows that the derivations start much before that and continue 
until the reconnection time. Furthermore, this local scaling appears to be 
a result of the global alignment that exists for almost all times.

Our conclusion is that the new scaling laws appear at all times for initially
orthogonal vortices and these scaling laws are probably tied to
the unique alignment of the Frenet-Serret frames that form early and 
continue through the reconnection period until to the end of each calculation. 
That these anomalous scaling laws are identical, about their respective
reconnection times, for all initial separations, implies that the
anomalous scaling laws could exist for initial vortices
with macroscopic initial separations extending to observable scales.

\section{Discussion \label{sec:conclude}}

These results leave us with several major questions. 
\ITM\item 
First, could the scaling laws shown here
be extended to the huge range of length scales in experiments? Because the new orthogonal
scaling laws appear at all times for initially orthogonal vortices and these scaling laws 
are tied to unique alignments that form early and continue through the reconnection 
period, it is possible that the anomalous scaling could apply to vortices on the
macroscopic, observable scales. 

However, what if the initial state is not strictly orthogonal? What seems to be true, 
based upon several additional curved configurations considered in
\cite{Rorai12} as well as cases from \cite{Zuccheretal2012}, is that 
the scaling of all quantum reconnection events should lie between the two extremes 
presented here. More work will be needed to determine when and for how long each
type of scaling dominates.

\item Second, can these cases be compared with the experiments using solid hydrogen markers?  
Improvements in both the experimental and numerical data sets will be needed before
that can be addressed properly.  Currently, a few isolated events in some of the 
experimental videos and the first experimental paper \citep{Bewleyetal_PNAS08} 
might be consistent with orthogonal scaling asymmetries described here. However,
in the best statistical analysis \citep{Paolettietal08}, the distributions of the
approaches and separations are clustered about the dimensional prediction, represented
by our anti-parallel case. 

\item Finally, how can the alignments quantified here for the orthogonal cases be used to
explain the anomalous reconnection scaling laws? The local swaps in the alignment 
of the Frenet-Serret frames for the orthogonal cases in figures \ref{fig:TNBD} and
\ref{fig:D-TNB} are probably too similar to the local swaps for the anti-parallel case 
in figure  \ref{fig:antipcurv} to explain the non-dimensional scaling laws.
So a better place to start might be to consider the large-scale alignments. However,
to use these alignments together with Biot-Savart to predict velocities could lead
nowhere since all of full Biot-Savart calculations find the dimensional, temporally 
symmetric scaling laws. 

\item[$\circ$] Nonetheless, Biot-Savart can be a useful place to start looking in the sense that
\eqref{eq:rho0v} provides us with a means to exactly determine the Gross-Pitaevskii
velocities, which could then be compared to the Biot-Savart predictions. Once the
differences have been identified, and from there the sources of these differences, we
should be on the road to explaining this new behaviour.
\ITN

\section*{Acknowledgements}
CR acknowledges support from the National Science Foundation, 
NSF-DMR Grant No.  0906109 and support of the Universit\'a di Trieste. 
RMK acknowledges support from the EU COST
Action program MP0806 ‘Particles in Turbulence’.
Discussions with C. Barenghi and M.E. Fisher have been appreciated.
Support with graphics from R. Henshaw is appreciated.
%\input{Bib25jul14}
%\newpage

%\newpage\null\thispagestyle{empty}\newpage
%\newpage
{}

\end{document}